\documentclass{article}[12pt]
\usepackage{graphicx,epsfig,color}
\usepackage{enumerate}
\usepackage{algorithm}
\usepackage[noend]{algpseudocode}
\usepackage{amsfonts}
\usepackage{amsmath}
\newcounter{savenumi}

\newtheorem{theoremfoo}{Theorem}
\newenvironment{theorem}{\pagebreak[1]\begin{theoremfoo}}{\end{theoremfoo}}

\newtheorem{propositionfoo}[theoremfoo]{Proposition}

\newtheorem{lemmafoo}[theoremfoo]{Lemma}
\newenvironment{lemma}{\pagebreak[1]\begin{lemmafoo}}{\end{lemmafoo}}
\newtheorem{conjecturefoo}[theoremfoo]{Conjecture}

\newtheorem{corollaryfoo}[theoremfoo]{Corollary}
\newenvironment{corollary}{\pagebreak[1]\begin{corollaryfoo}}{\end{corollaryfoo}}
\newtheorem{exercisefoo}{Exercise}

\newtheorem{openfoo}[theoremfoo]{Question}

\newtheorem{nttn}[theoremfoo]{Notation}

\newtheorem{dfntn}[theoremfoo]{Definition}
\newenvironment{definition}{\pagebreak[1]\begin{dfntn}\rm}{\end{dfntn}}

\newenvironment{proof}
    {\pagebreak[1]{\narrower\noindent {\bf Proof:\quad\nopagebreak}}}{\QED}

\newcommand{\poly}{{\rm poly}}
\def\nre.{$n$\/-r.e.}

\newtheorem{factfoo}[theoremfoo]{Fact}

\newcommand{\squeeze}{
\textwidth 6in
\textheight 8.8in
\oddsidemargin 0.2in
\topmargin -0.4in
}

\newtheorem{propertyfoo}[theoremfoo]{Property}

\makeatletter

\def\@makechapterhead#1{ \vspace*{50pt} { \parindent 0pt \raggedright 
 \ifnum \c@secnumdepth >\m@ne \huge\bf \@chapapp{} \thechapter. \par 
 \vskip 20pt \fi \Huge \bf #1\par 
 \nobreak \vskip 40pt } }

\def\@sect#1#2#3#4#5#6[#7]#8{\ifnum #2>\c@secnumdepth
     \def\@svsec{}\else 
     \refstepcounter{#1}\edef\@svsec{\csname the#1\endcsname.\hskip 1em }\fi
     \@tempskipa #5\relax
      \ifdim \@tempskipa>\z@ 
        \begingroup #6\relax
          \@hangfrom{\hskip #3\relax\@svsec}{\interlinepenalty \@M #8\par}
        \endgroup
       \csname #1mark\endcsname{#7}\addcontentsline
         {toc}{#1}{\ifnum #2>\c@secnumdepth \else
                      \protect\numberline{\csname the#1\endcsname}\fi
                    #7}\else
        \def\@svsechd{#6\hskip #3\@svsec #8\csname #1mark\endcsname
                      {#7}\addcontentsline
                           {toc}{#1}{\ifnum #2>\c@secnumdepth \else
                             \protect\numberline{\csname the#1\endcsname}\fi
                       #7}}\fi
     \@xsect{#5}}

\def\@begintheorem#1#2{\it \trivlist \item[\hskip \labelsep{\bf #1\ #2.}]}

\def\@opargbegintheorem#1#2#3{\it \trivlist
      \item[\hskip \labelsep{\bf #1\ #2\ (#3).}]}

\makeatother

\newif\ifshortconferences
\shortconferencesfalse
\newif\ifmediumconferences
\mediumconferencesfalse

\def\ending#1{{\count1=#1\relax
\count2=\count1
\divide\count2 by 100
\multiply\count2 by 100
\advance\count1 by -\count2
\ifnum\count1=11
th%
\else \ifnum\count1=12
th%
\else \ifnum\count1=13
th%
\else 
\count2=\count1
\divide\count1 by 10
\multiply\count1 by 10
\advance\count2 by -\count1
\ifnum\count2=1
st%
\else \ifnum\count2=2
nd%
\else \ifnum\count2=3
rd%
\else th%
\fi\fi\fi\fi\fi\fi
}}

\def\Proceedingsofthe{\ifshortconferences Proc.\else\ifmediumconferences Proc.\else Proceedings of the\fi\fi}

\newcounter{confnum}

\def\conf#1#2{%
\setcounter{confnum}{#2}%
\addtocounter{confnum}{-\csname #1zero\endcsname}%
\ifnum\value{confnum}=1%
\expandafter\ifx\csname #1One\endcsname\relax%
\Proceedingsofthe\ \arabic{confnum}\ending{\value{confnum}}\ \csname #1name\endcsname%
\else \csname #1One\endcsname\fi%
\else%
\Proceedingsofthe\
\arabic{confnum}\ending{\value{confnum}}\ \csname #1name\endcsname\fi}

\def\qsym{\vrule width0.7ex height0.9em depth0ex}

\newif\ifqed\qedtrue

\def\noqed{\global\qedfalse}

\def\qed{\ifqed{\penalty1000\unskip\nobreak\hfil\penalty50
\hskip2em\hbox{}\nobreak\hfil\qsym
\parfillskip=0pt \finalhyphendemerits=0\par\medskip}\fi\global\qedtrue}

\makeatletter
\def\eqnqed{\noqed
	\def\@tempa{equation}
	\ifx\@tempa\@currenvir\def\@eqnnum{\qsym}%
	\addtocounter{equation}{-1}\else%
    \def\@@eqncr{\let\@tempa\relax
    \ifcase\@eqcnt \def\@tempa{& & &}\or \def\@tempa{& &}%
      \else \def\@tempa{&}\fi
     \@tempa {\def\@eqnnum{{\qsym}}\@eqnnum}
     \global\@eqnswtrue\global\@eqcnt\z@\cr}\fi}

\def\eqnlabel#1#2{\if@filesw {\let\thepage\relax%
   \def\protect{\noexpand\noexpand\noexpand}%
   \edef\@tempa{\write\@auxout{\string
      \newlabel{#2}{{{#1}}{\thepage}}}}%
   \expandafter}\@tempa%
   \if@nobreak \ifvmode\nobreak\fi\fi\fi%
	\def\@tempa{equation}
	\ifx\@tempa\@currenvir\def\theequation{{#1}}%
	\addtocounter{equation}{-1}\else%
    \def\@@eqncr{\let\@tempa\relax
    \ifcase\@eqcnt \def\@tempa{& & &}\or \def\@tempa{& &}%
      \else \def\@tempa{&}\fi
     \@tempa {\def\@eqnnum{{#1}}\@eqnnum}
     \global\@eqnswtrue\global\@eqcnt\z@\cr}\fi}

\makeatother

\def\QED{\qed}

\makeatother

\squeeze

\newcounter{example}[section]
\newenvironment{example}[1][]{\refstepcounter{example}\par\medskip
    \noindent \textbf{Example~\theexample. #1} \rmfamily}{\medskip}

\newcommand{\bigO}{{\rm O}}
\newcommand{\prob}{{\rm Prob}}

\newcommand{\Thick}{{\rm minThickness}}
\newcommand{\ThickMax}{{\rm maxThickness}}

\newcommand{\randomElm}{{\rm RandomElement}}

\newcommand{\currentT}{{\rm currentThickness}}

\newcommand{\setCount}{ m}

\newcommand{\complexityLatticePoint}{\bigO\left(\frac{1}{\beta}d^{\frac{11}{2}}l^3\left(\frac{2d^{\frac{3}{2}}}{\beta}+l|\lambda|\right)^3\right)}

\begin{document}

    \date{}

    \title{Approximate Set Union via Approximate Randomization
        \thanks{This research is supported in part by National Science
            Foundation Early Career Award 0845376 and Bensten Fellowship of the
            University of Texas - Rio Grande Valley.} }

    \author{
        Bin Fu$^1$, Pengfei Gu$^1$, and Yuming Zhao$^2$
        \\\\
        $^1$Department of Computer Science\\
        University of Texas - Rio Grande Valley, Edinburg, TX 78539, USA\\
        \\ \\
        $^2$School of Computer Science\\
        Zhaoqing University, Zhaoqing, Guangdong 526061, P.R. China\\
        \\
        \\{\bf}
    } \maketitle

    \begin{abstract}
        We develop a randomized approximation algorithm for the size of set union problem   $\arrowvert A_1\cup A_2\cup...\cup A_{\setCount}\arrowvert$,
        which is given a list of sets $A_1,\,...,\,A_{\setCount}$ with approximate set size $m_i$ for
        $A_i$ with $m_i\in \left((1-\beta_L)|A_i|,\, (1+\beta_R)|A_i|\right)$, and biased
        random
        generators with $\prob\left(x=\randomElm(A_i)\right)\in
        \left[{1-\alpha_L\over |A_i|},\, {1+\alpha_R\over |A_i|}\right]$ for each input
        set $A_i$ and element $x\in A_i,$ where $i=1,\, 2,\, ...,\, \setCount$. The approximation ratio for
        $\arrowvert A_1\cup A_2\cup...\cup A_m\arrowvert$ is in the range
        $\left[(1-\epsilon)(1-\alpha_L)(1-\beta_L),\,
        (1+\epsilon)(1+\alpha_R)(1+\beta_R)\right]$ for any $\epsilon\in (0,\,1)$,
        where $\alpha_L,\, \alpha_R,\, \beta_L,\,\beta_R\in (0,\,1)$. The complexity
        of the algorithm is measured by both
        time complexity and round complexity. The algorithm is allowed to
        make multiple membership queries and get random elements from the
        input sets in one round.  Our algorithm makes adaptive accesses to input sets with multiple rounds.
        Our algorithm gives an approximation scheme with $\bigO(\setCount\cdot(\log \setCount)^{\bigO(1)})$ running time and $\bigO(\log \setCount)$ rounds,
        where $\setCount$ is the number of sets. Our algorithm can handle input sets that can generate random elements with bias,
        and its approximation ratio depends on the
        bias. Our algorithm gives a flexible tradeoff with time complexity $\bigO\left(\setCount^{1+\xi}\right)$ and round complexity $\bigO\left({1\over \xi}\right)$ for any $\xi\in(0,\,1)$.
        We prove that our algorithm runs sublinear in time under certain condition that each element in $A_1\cup A_2\cup...\cup A_{\setCount}$ belongs to $\setCount^a$ for
        any fixed $a>0$. A $\bigO\left(r(r+l|\lambda|)^3l^3d^4\right)$ running time dynamic programming algorithm is proposed to deal with an interesting problem in number
        theory area that is to count the number of lattice points in a $d-$dimensional ball $B_d(r,\, p,\, d)$ of radius $r$ with center at $p\in D(\lambda,\, d,\, l)$,
        where $D(\lambda,\,d,\, l)=\{ (x_1,\,\cdots,\, x_d) : (x_1,\,\cdots,\, x_d)$ with $x_k=i_k+j_k\lambda$ for an integer $j_k\in [-l,\,l]$, and another arbitrary integer $i_k$ for $k=1,\, 2,\, ...,\, d\}$.
        We prove that it is $\#$P-hard to count the number of lattice points in a set of balls, and we also show that there is no polynomial time algorithm to approximate the number
        of lattice points in the intersection of $n$-dimensional  balls unless P=NP.
    \end{abstract}


    \section{Introduction}
    Computing the cardinality of set union is a basic algorithmic
    problem that has a simple and natural definition. It is related to
    the following problem: given a list of sets $A_1,\,...,\,A_{\setCount}$
    with set size $|A_i|$, and random generators  $\randomElm(A_i)$ for
    each input set $A_i$, where $i=1,\, 2,\, ...,\, \setCount,$ compute
    $\arrowvert A_1\cup A_2\cup...\cup A_{\setCount}\arrowvert$. This
    problem is $\#$P-hard if each set is $0,\,1$-lattice points in a high
    dimensional cube \cite{Valiant}. Karp, Luby, and Madras
    \cite{KarpLubyMadras89} developed a $(1+\epsilon)$-randomized
    approximation algorithm to improve the runnning time for
    approximating the number of distinct elements in the union
    $A_1\cup\cdots\cup A_{\setCount}$ to linear
    $\bigO((1+\epsilon)\setCount/\epsilon^2)$ time. Their algorithm is
    based on the input that provides the exact size of each set and an
    uniform random element generator of each set. Bringmann and
    Friedrich~\cite{BringmannFriedrich10} applied Karp, Luby, and
    Madras' algorithm in deriving approximate algorithm for high
    dimensional geometric object with uniform random sampling. They also
    proved that it is \#P-hard to compute the volume of the intersection
    of high dimensional boxes, and showed that there is no polynomial
    time $2^{d^{1-\epsilon}}$-approximation unless NP=BPP. In the
    algorithms mentioned above, some of them were based on random
    sampling, and some of them provided exact set sizes when approximating
    the cardinalities of multisets of data and some of them dealt with
    two multiple sets. However, in realty, it is really hard to give an
    uniform sampling or exact set size especially when deal with high
    dimensional problems.

    A similar problem has been studied in the streaming model: given a
    list of elements with multiplicity, count the number of distinct
    items in the list.  This problem has a more general format to
    compute frequency moments $F_k=\sum\limits_{i=1}^{m}n_{i}^{k}$,
    where $n_i$ denotes the number of occurrences of $i$ in the sequence.
    This problem has received a lot of attention in the field of
    streaming
    algorithms~\cite{AlonMatiasSzegedy96,BK02,YossefJayramKumarSivakumar02,B18,FFGM07,FlajoletMartin85,GangulyGarofalakisRastogi04,Gibbons01,GibbonsTirthapura01,HaasNaughtonSeshadriStokes95,HTY14,KNM10}.

    {\bf Motivation:} The existing approximate set union
    algorithm~\cite{KarpLubyMadras89} needs each input set has a uniform
    random generator. In order to have approximate set union algorithm
    with broad application, it is essential to have algorithm with
    biased random generator for each input set, and see how
    approximation ratio depends on the bias. In this paper, we propose a
    randomized approximation algorithm to approximate the size of set
    union problem by extending the model used
    in~\cite{KarpLubyMadras89}. In order to show why approximate
    randomization method is useful, we generalize the algorithm that was
    designed by Karp, Luby, and Madras \cite{KarpLubyMadras89} to an
    approximate randomization algorithm. A natural problem that counting
    of lattice points in d-dimensional ball is discussed to support the
    useful of approximate randomization algorithm. In our algorithm,
    each input set $A_i$ is a black box that can provide its size
    $|A_i|$, generate a random element $\randomElm(A_i)$ of $A_i$, and
    answer the membership query $(x\in A_i?)$ in $O(1)$ time. Our
    algorithm can handle input sets that can generate random elements
    with bias with $\prob(x=\randomElm(A_i))\in \left[{1-\alpha_L\over
        |A_i|},\,{1+\alpha_R\over |A_i|}\right]$ for each input set $A_i$ and approximate set size $m_i$ for $A_i$ with
    $m_i\in [(1-\beta_L)|A_i|,\, (1+\beta_R)|A_i|]$.

    As the communication complexity is becoming important in distributed
    environment, data transmission among variant machines may be more
    time consuming than the computation inside a single machine. Our
    algorithm complexity is also measured by the number of rounds. The
    algorithm is allowed to make multiple membership queries and get
    random elements from the input sets in one round. Our algorithm
    makes adaptive accesses to input sets with multiple rounds.  The
    round complexity is related a distributed computing complexity if
    input sets are stored in a distributed environment, and the number
    of rounds indicates the complexity of interactions between a central
    server, which runs the algorithm to approximate the size of set
    union, and clients, which save one set each.

    Computation  via bounded queries to another set has been well
    studied in the field of structural complexity theory. Polynomial
    time truth table reduction has a parallel way to access oracle with
    all queries to be provided in one round~\cite{BussHay-structures}.
    Polynomial time Turing reduction has a sequential way to access
    oracle by providing a query and receiving an answer in one
    round~\cite{Cook-NP-complete}. The constant-round truth table
    reduction (for example, see~\cite{FortnowReingold-tr}) is between
    truth table reduction, and Turing reduction. Our algorithm is
    similar to a bounded round truth table reduction to input sets to
    approximate the size set union. Karp, Luby, and Madras
    \cite{KarpLubyMadras89}'s algorithm runs like a Turing reduction
    which has the number of adaptive queries proportional to the time.

    We design approximation scheme for the number of lattice points in a
    $d$-dimensional ball with its center in $D(\lambda,\, d,\, l)$, where
    $D(\lambda,\,d,\, l)$ to be the set points $p_d=(x_1,\,\cdots,\, x_d)$ with
    $x_i=i+j\lambda$ for an integer $j\in [-l,\, l]$,  another arbitrary
    integer $i$, and an arbitrary real number $l$. It returns an
    approximation in the range $[{(1-\beta) C(r,\, p,\, d)},\,{(1+\beta) C(r,\,
        p,\, d)}]$ in a time $\poly\left(d,\, {1\over \beta},\, |l|,\,|\lambda|\right)$, where
    $C(r,\, p,\,d)$ is the number of lattice points in a $d$-dimensional
    ball with radius $r$ and center $p\in D(\lambda,\, d,\, l)$. We also
    show how to generate a random lattice point in a $d$-dimensional
    ball with its center at $D(\lambda,\, d,\, l)$. It generates each
    lattice point inside the ball with a probability in $\left[{1-\alpha\over
        C(r, p, d)},\,{1+\alpha\over C(r, p, d)}\right]$ in a time $\poly\left(d,\,{1\over
        \alpha},\, |l|,\,|\lambda|,\,\log r\right)$, where the $d$-dimensional ball has
    radius $r$ and center $p\in D(\lambda,\, d,\, l)$. Without the condition
    that a ball center is inside $D(\lambda,\,d,\, l)$, counting the number
    of lattice points in a ball may have time time complexity that
    depends on dimension number $d$ exponentially even the radius is as
    small as $d$. Counting the number of lattice points inside a four
    dimensional ball efficiently implies an efficient algorithm to
    factorize the product of two prime numbers ($n=pq$) as
    $C(\sqrt{n},\,(0,\,...,\,0),\, 4)-C(\sqrt{n-1},\,(0,...,0),\, 4)=8(1+pA+q+n)$
    (see~\cite{AndrewsEkhadZeilberger93, Jacobi69}). Therefore, a fast
    exact counting lattice points inside a four dimensional ball implies
    a fast algorithm to crack RSA public key system.

    This gives a natural example to apply our approximation scheme to
    the number of lattice points in a list of balls. We prove that it is
    $\#$P-hard to count the number of lattice points in a set of balls,
    and we also show that there is no polynomial time algorithm to
    approximate the number of lattice points in the intersection
    n-dimensional balls unless P=NP.  We found that it is an elusive
    problem to develop a $\poly\left(d,\,{1\over \epsilon}\right)$ time
    $(1+\epsilon)$-approximation algorithm for the number of lattice
    points of $d$-dimensional ball with a small radius. We are able to
    handle the case with ball centers in $D(\lambda,\, d,\, l)$, which can
    approximate an arbitrary center by adjusting parameters $\lambda$
    and $l$. This is our main technical contributions about lattice
    points in a high dimensional ball.

    It is a classical problem in analytic number theory for counting the
    number of lattice points in d-dimensional ball, and has been studied
    in a series of
    articles~\cite{SD91,Beck89,KR67,Chen63,KI67,hafner81,Brown99,hueley93,Mazo90,meyer13,Szeg26,Tsang00,AM79,Vi63,Walfisz63,Walfisz57,yudin67}
    in the field of number theory. Researchers are interested in both
    upper bounds and lower bounds for the error term
    $E_d(r)=N_d(r)-\pi^{\frac{d}{2}}\Gamma(\frac{1}{2}d+1)^{-1}r^d,$
    where $N_d(r)=\#\{x\in \mathbb{Z}^d: |x|\leq r\}$  is the number of
    lattice points inside a sphere of radius $r$ centered at the origin
    and $\pi^{\frac{d}{2}}\Gamma(\frac{1}{2}d+1)^{-1}r^d$ (where
    $\Gamma(.)$ is Gamma Function) is the volume of a
    $d-dimensional$ sphere of radius $r$. When $d=2$, the problem is
    called \lq\lq Gauss Circle Problem\rq\rq; Gauss proved that $E_2(r)\leq r$.
    Gauss's bound was improved in papers~\cite{KI67, hafner81, hueley93}. Walfisz \cite{Walfisz57} showed
    that $E_d(r)=\Omega_{\pm}(r^{d-2})$ and $E_d(r)\leq r^{d-2}$, where
    $f(x)=\Omega_{+}(F(x))(f(x)=\Omega_{-}(F(x)))$ as
    $x\rightarrow\infty$ if there exist a sequence
    $\{x_n\}\rightarrow\infty$ and a positive number $C$, such that for
    all $n\geq1,$ $f(x_n)>C|F(x_n)|$ ($f(x_n)<-C|F(x_n)|$). Most of the
    above results focus on the ball centered at the origin, and few
    papers worked on variable centers but also consider fixed dimensions
    and radii going to infinity \cite{Beck89,KR67,AM79,yudin67}.

    {\bf Our Contributions:} We have the following contributions to
    approximate the size of set union. 1. It has constant number of
    rounds to access the input sets. This reduces an important
    complexity in a distributed environment where each set stays a
    different machine. It is in contrast to the existing algorithm that
    needs $\Omega(\setCount)$ rounds in the worst case. 2. It handles
    the approximate input set sizes and biased random sources. The
    existing algorithms assume uniform random source from each set. Our
    approximation ratio depends on the approximation ratio for the input
    set sizes and bias of random generator of each input set. The
    approximate ratio for $|A_1\cup A_2\cup \cdots\cup A_{\setCount}|$
    is controlled in the range in $\left[(1-\epsilon)(1-\alpha_L)(1-\beta_L),\,
    (1+\epsilon)(1+\alpha_R)(1+\beta_R)\right]$ for any $\epsilon\in (0,\,1)$,
    where $\alpha_L,\, \alpha_R,\, \beta_L,\,\beta_R\in (0,\,1)$. 3. It runs in
    sublinear time when each element belongs to at least $\setCount^a$
    sets for any fixed $a > 0.$ We have not seen any sublinear results
    about this problem. 4. We show a tradeoff between the number of
    rounds, and the time complexity. It takes $\log m$ rounds with
    time complexity $\bigO\left(\setCount(\log
    \setCount)^{O(1)}\right)$, and takes
    $\bigO\left(\frac{1}{\xi}\right)$ rounds, with a time complexity
    $\bigO\left(\setCount^{1+\xi}\right)$. We still maintain the time
    complexity nearly linear time in the classical model. Our algorithm
    is based on a new approach that is different from that in
    \cite{KarpLubyMadras89}. 5. We identify two additional parameters
    $z_{min}$ and $z_{max}$ that affect both the complexity of rounds
    and time, where $z_{min}$ is the least number of sets that an
    element belongs to, and $z_{max}$ is the largest number of sets that
    an element belongs to.

    Our algorithm developed in the randomized model only accesses a
    small number of elements from the input sets. The algorithm
    developed in the  streaming model algorithm accesses all the
    elements from the input sets. Therefore, our algorithm is
    incomparable with the results in the streaming
    model~\cite{AlonMatiasSzegedy96,BK02,YossefJayramKumarSivakumar02,B18,FFGM07,FlajoletMartin85,GangulyGarofalakisRastogi04,Gibbons01,GibbonsTirthapura01,
        HaasNaughtonSeshadriStokes95,HTY14,KNM10}.

    {\bf Organization:} The rest of paper is organized as follows. In
    Section~\ref{model-sec}, we define the computational model and
    complexity. Section~\ref{overview-sec} presents some theorems that
    play an important role in accuracy analysis. In Section~\ref{alg},
    we give a randomized approximation algorithm to approximate the size
    of set union problem; time complexity and round complexity also
    analysis in Section~\ref{alg}. Section~\ref{random-lattice-section}
    discusses a natural problem that counting of lattice points in high
    dimensional balls to support the useful of approximation randomized
    algorithm. An application of high dimensional balls in Maximal
    Coverage gives in Section~\ref{ser}.
    In Section~\ref{concl}, we summarize with conclusions.

\section{Computational Model and Complexity}\label{model-sec}
In this section, we show our model of computation, and the
definition of complexity.

\subsection{Model of Randomization}

\begin{definition}
    Let $A$ be a set of elements.
    \begin{enumerate}
        \item
        A {\it $\alpha$-biased random generator} for set $A$ is a generator
        that each element in $A$ is generated with probability in the range
        $\left[{1-\alpha\over |A|},\,{1+\alpha\over |A|}\right]$.
        \item
        A {\it $(\alpha_L,\alpha_R)$-biased random generator} for set $A$ is
        a generator that each element in $A$ is generated with probability
        in the range $\left[{1-\alpha_L\over |A|},\,{1+\alpha_R\over
            |A|}\right]$.
    \end{enumerate}
\end{definition}

\begin{definition}\label{input-list-def}
    Let $L$ be a list of sets $A_1,\,A_2,\,\cdots,\, A_{\setCount}$ such that
    each supports the following operations:
    \begin{enumerate}
        \item
        The size of $A_i$ has an approximation $m_i\in \left[(1-\beta_L)|A_i|,\,
        (1+\beta_R)|A_i|\right]$ for $i=1,\,2,\,\cdots,\,\setCount$. Both $M=\sum\limits_{i=1}^\setCount m_i$ and $m$
        are part of the input.
        \item
        Function \randomElm$(A_i)$ returns a $(\alpha_L,\,\alpha_R)$-biased
        approximate random element $x$ from $A_i$ for $i=1,\,2,\,\cdots,\,
        \setCount$.
        \item
        Function query$(x,\, A_i$) function returns $1$ if $x\in A_i$, and $0$
        otherwise.
    \end{enumerate}
\end{definition}

\begin{definition}\label{app-list-def}
    For a list $L$ of sets $A_1,\,A_2,\,\cdots,\,A_{\setCount}$ and real numbers
    $\alpha_L,\,\alpha_R,\,\beta_L,\,\beta_R\in [0,\,1)$, it is called
    $\left((\alpha_L,\,\alpha_R),\,(\beta_L,\,\beta_R)\right)$-list if each set $A_i$
    is associated with a number $m_i$ with $(1-\beta_L)|A_i|\le m_i\le
    (1+\beta_R)|A_i|$ for $i=1,\,2,\,\cdots,\,\setCount$, and the set $A_i$
    has a $(\alpha_L,\,\alpha_R)$-biased random generator
    \randomElm($A_i$).
\end{definition}

\begin{definition}\label{model-def}
    The model of randomized computation for our algorithm is defined below:
    \begin{enumerate}
        \item
        The input is a list $L$ defined in Definition~\ref{input-list-def}.
        \item
        It allows all operations defined in Definition~\ref{input-list-def}.
    \end{enumerate}
\end{definition}

\subsection{Round and Round Complexity}


The {\it round complexity} is the total number of rounds used in the
algorithm. Our algorithm has several rounds to access input sets. At each round, the
algorithm send multiple requests to random generators, and
membership queries, and receives the answers from them.

Our algorithm is considered as a client-server interaction (see Fig. 1). The
algorithm is controlled by the server side, and each set is a
client. In {\it one round}, the server asks some questions to
clients which are selected.
\begin{figure}[H]
    \centering
    \includegraphics[width=0.5\textwidth]{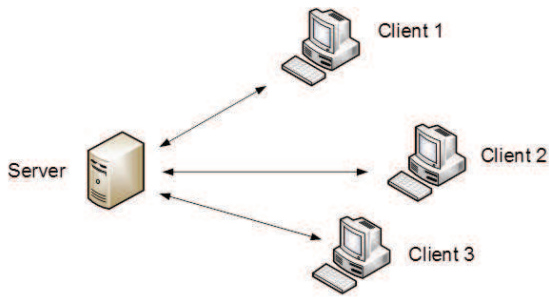} 
    \caption{Client-server Interaction}
\end{figure}

The  parameters $\setCount,\, \epsilon,\, \gamma$ may be used to
determine the time complexity and round complexity, where $\epsilon$ controls
the accuracy of approximation, $\gamma$ controls the failure
probability, and $\setCount$ is the number of sets.

\section{Preliminaries}\label{overview-sec}
During the accuracy analysis, Hoeffiding Inequality
\cite{Hoeffding63} and Chernoff Bound
(see~\cite{MotwaniRaghavan00}) play an important role. They show how the number
of samples determines the accuracy of approximation.




\begin{theorem}[\cite{Hoeffding63}]\label{hoeffiding-theorem}
    Let $X_1,\,\ldots,\, X_{\setCount}$ be $\setCount$ independent
    random variables in $[0,\,1]$ and $X=\sum\limits_{i=1}^{\setCount} X_i$.

    i. If $X_i$ takes $1$ with probability at most $p$ for $i=1,\,\ldots,\,
    \setCount$, then for any $\epsilon>0$,
    $\Pr(X>p\setCount+\epsilon \setCount)<e^{-{\epsilon^2
            \setCount\over 2}}$.

    ii. If  $X_i$ takes $1$ with probability at least $p$ for
    $i=1,\,\ldots ,\, \setCount$, then for any $\epsilon>0$,
    $\Pr(X<p\setCount-\epsilon \setCount)<e^{-{\epsilon^2
            \setCount\over 2}}$.
\end{theorem}




\begin{theorem}\label{chernoff1-theorem}
    Let $X_1,\,\ldots,\, X_{\setCount}$ be $\setCount$ independent random $0$-$1$ variables,
    where $X_i$ takes $1$ with probability at least $p$ for $i=1,\,\ldots
    ,\, \setCount$. Let $X=\sum\limits_{i=1}^\setCount X_i$, and $\mu=E[X]$. Then for any
    $\delta>0$,
    $\Pr(X<(1-\delta)p\setCount)<e^{-{1\over 2}\delta^2 p\setCount}$.
\end{theorem}

\begin{theorem}\label{ourchernoff2-theorem}
    Let $X_1,\,\ldots,\, X_\setCount$ be $\setCount$ independent random $0$-$1$ variables,
    where $X_i$ takes $1$ with probability at most $p$ for $i=1,\,\ldots,\,
    \setCount$. Let $X=\sum\limits_{i=1}^\setCount X_i$. Then for any $\delta>0$,
    $\Pr(X>(1+\delta)p\setCount)<\left[{e^{\delta}\over
        (1+\delta)^{(1+\delta)}}\right]^{p\setCount}$.
\end{theorem}

Define $g_1(\delta)=e^{-{1\over 2}\delta^2}$ and
$g_2(\delta)={e^{\delta}\over (1+\delta)^{(1+\delta)}}$. Define
$g(\delta)=\max\left(g_1(\delta),\, g_2(\delta)\right)$. We note that
$g_1(\delta)$ and $g_2(\delta)$ are always strictly less than $1$
for all $\delta>0$. It is trivial for $g_1(\delta)$. For
$g_2(\delta)$, this can be verified by checking that the function
$f(x)=(1+x)\ln (1+x)-x$ is increasing and $f(0)=0$. This is because
$f'(x)=\ln (1+x)$ which is strictly greater than $0$ for all $x>0$.

We give a bound for ${e^{\delta}\over (1+\delta)^{(1+\delta)}}$. Let
$u(x)={e^x\over (1+x)^{(1+x)}}$. We consider the case $x\in [0,\,1]$.
We have
\begin{eqnarray*}
    \ln u(x)&=&x-(1+x)\ln (1+x)
    \le x-(1+x)(x-{x^2\over 2})
    =x-(x+{x^2\over 2}-{x^3\over 3})
    \le-{x^2\over 6}.
\end{eqnarray*}
Therefore,
\begin{eqnarray}
u(x)\le e^{-{x^2\over 6}}\label{u(x)-ineqn}
\end{eqnarray}
for all $x\in [0,\,1]$.  We let
\begin{eqnarray}
g^*(x)= e^{-{x^2\over 6}}.\label{g*(x)-eqn}
\end{eqnarray}
We have $g(x)\le g^*(x)$ for all $x\in [0,\,1]$.

A well known fact, called union bound, in probability theory is the
inequality
$$\Pr(E_1\cup E_2 \ldots \cup E_\setCount)\le
\Pr(E_1)+\Pr(E_2)+\ldots+\Pr(E_\setCount),$$ where $E_1,\,E_2,\,\ldots,\, E_{\setCount}$ are
$\setCount$ events that may not be independent. In the analysis of our
randomized algorithm, there are multiple events such that the
failure from any of them may fail the entire algorithm. We often
characterize the failure probability of each of those events, and
use the above inequality to show that the whole algorithm has a
small chance to fail after showing that each of them has a small
chance to fail.

\section{Algorithm Based on Adaptive Random Samplings
}\label{alg}

In this section, we develop a randomized algorithm for the size of
set union when the approximate set sizes and biased random
generators are given for the input sets. We give some definitions
before the presentation of the algorithm. The algorithm developed in
this section has an adaptive way to access the random generators
from the input sets. All the random elements from input sets are
generated in the beginning of the algorithm, and the number of
random samples is known in the beginning of the algorithm. The
results  in this section show a tradeoff between the time complexity
and the round complexity.

\begin{definition}\label{partition-thickness-defn}
    Let $L=A_1,A_2,\cdots, A_{\setCount}$ be a list of finite sets.
    \begin{enumerate}
        \item
        For an element $x$, define $T(x,L)=\big|\{i: 1\le i\le \setCount$  and
        $x \in A_i\}\big|$.
        \item
        For an element $x$, and a subset of indices with multiplicity
        $H$ of $\{1,2,\cdots, \setCount\}$, define $S(x, H)=\big|\{i: i\in
        H$ and $x\in A_i\}\big|$.
        \item
        Define $\Thick(L)=\min\{T(x,L): x\in A_1\cup A_2\cup\cdots\cup
        A_{\setCount}\}$.
        \item
        Define $\ThickMax(L)=\max\{T(x,L): x\in A_1\cup A_2\cup\cdots\cup
        A_{\setCount}\}$.
        \item
        Let $W$ be a subset with multiplicity of $A_1\cup\cdots\cup
        A_{\setCount}$, define $F(W, h,s)={s\over h}\sum\limits_{x\in W}{1\over
            T(x, L)}$, and $F'(W)=\sum\limits_{x\in W}{1\over T(x, L)}={h\over s}F(W,
        h,s)$.
        \item\label{partition-def}
        For a $\delta \in (0,1)$, partition $A_1\cup A_2\cup\cdots\cup
        A_{\setCount}$ into $A_1',\cdots, A_k'$ such that $A_i'=\{x: x\in
        A_1\cup A_2\cup\cdots\cup A_{\setCount}$ and $T(x, L)\in
        [(1+\delta)^{i-1}, (1+\delta)^i)\}$ where $i=1, 2, ..., k.$
        Define $v(\delta, z_1, z_2, L)=k$, which is the
        number of sets in the partition under the condition that $z_1\le
        T(x,L)\le z_2$.
    \end{enumerate}
\end{definition}

\subsection{Overview of Algorithm}

We give an overview of the algorithm. For a list $L$ of input sets
$A_1,\cdots, A_{\setCount}$, each set $A_i$ has an approximate size
$m_i$ and a random generator. It is easy to see that $|A_1\cup
A_2\cup\cdots \cup A_{\setCount}|=\sum\limits_{i=1}^{\setCount}\sum\limits_{x\in
    A_i} {1\over T(x,L)}$. The first phase of the algorithm generates a
set $R_1$ of sufficient random samples from the list of input sets. The
set $R_1$ has the property that ${m_1+\cdots+m_{\setCount}\over
    |R_1|}\cdot \sum\limits_{x\in R_1} {1\over T(x,L)}$ is close to
$\sum\limits_{i=1}^{\setCount}\sum\limits_{x\in A_i} {1\over T(x,L)}$. We will use
the variable $sum$ with initial value zero to approximate it. Each
stage $i$ removes the set $V_i$ of  elements from $R_i$ that each
element $x \in V_i$ satisfies  $T(x,L)\in \left[{T_i\over
    4f_1(\setCount)}, T_i\right]$, and all elements $x\in R_i$ with $T(x,L)\in
\left[{T_i\over f_1(\setCount)}, T_i\right]$ are in $V_i$, where
$T_i=\max\{T(x,L): x\in R_i\}$ and $f_1(\setCount)$ is a function at
least $8$, which will determine the number of rounds, and the trade
off between the running time and the number of rounds. In phase $i$, we
choose a set $H_i$ of $u_i$ (to be large enough) of indices from
$1,\cdots, \setCount$, and use ${S(x, H_i)\setCount\over u_i}$ to
approximate $T(x, L)$. It is accurate enough if $u_i$ is large
enough. The elements left in $R_i-V_i$ will have smaller $T(x,L)$.
The set $R_{i+1}$ will be built for the next stage $i+1$. When
$R_i-V_i$ is shrinked to $R_{i+1}$ by random sampling in $R_i-V_i$,
each element in $R_{i+1}$ will have its weight to be scaled by a
factor ${|R_i-V_i|\over h_{i+1}}$. When an element $x$ is put into
$V_i$, it is removed from $R_i$, and an approximate value of
${1\over T(x,L)}$ multiplied by its weight is added to $sum$.
Finally, we will prove that $sum\cdot (m_1+\cdots +m_{\setCount})$
is close to $\sum\limits_{i=1}^{\setCount}\sum\limits_{x\in A_i} {1\over T(x,L)}$,
which is equal to $|A_1\cup A_2\cup\cdots \cup A_{\setCount}|$.

\begin{example}
    Let $L$ be a list of $10$ sets $A_1,\, A_2,\,\cdots,\, A_{10},$ where $A_i=B_i\cup C$ with $|C|=1000$ and $|B_i|=100$ for $i=1,\, 2,\, \cdots,\, 10.$ In the beginning of the algorithm, we generate a set $R_1$ of $h_1=220$ random samples from list $L,$ where there are $200$ random samples with higher thickness $T(x, \, L),$ namely, these $200$ random samples locate in $C$ and $20$ random samples with lower thickness $T(x, \, L)$, say, these $20$ random samples locate in $B_i.$ At the first round, we only need select sets $A_1,\, A_3,$ and $A_6$  to approximate the thickness $T(x, \, L)$ of the $200$ random samples locating at $C.$ Then at the second round, we have to select all the sets $A_1,\, A_2,\,\cdots,\, A_{10},$ to approximate the thickness $T(x, \, L)$ of the $20$ random samples coming from $B_i$ (See Fig. 2).
\end{example}
\begin{figure}[H]
    \centering
    \includegraphics[width=0.5\textwidth]{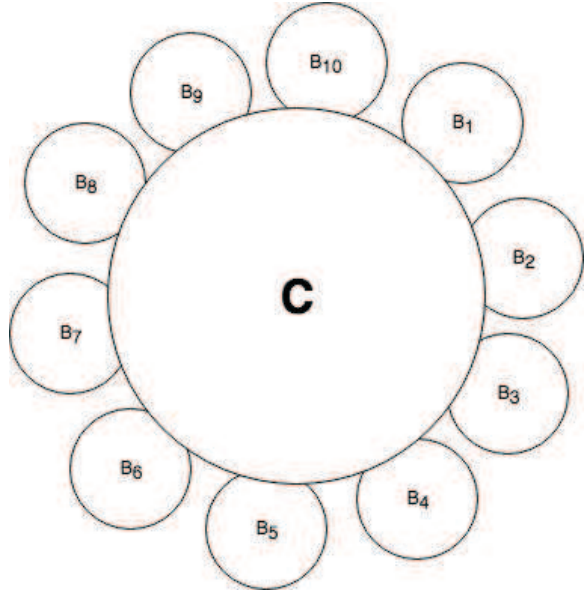} 
    \caption{Set Union of Ten Sets}
\end{figure}

\subsection{Algorithm Description}


Before giving the algorithm, we define an operation that selects a
set of random elements from a list $L$ of sets $A_1,\cdots,
A_{\setCount}$. We always assume $\setCount\ge 2$ throughout the
paper.

\begin{definition}\label{random-choice-def}
    Let $L$ be a list of $\setCount$ sets $A_1, A_2,\cdots,
    A_{\setCount}$ with $m_i\in [(1-\beta_L)|A_i|, (1+\beta_R)|A_i|]$
    and $(\alpha_L, \alpha_R)$-biased random generator \randomElm($A_i$)
    for $i=1,2,\cdots, \setCount$, and $M=m_1+m_2+\cdots+m_{\setCount}$.
    A {\it random choice} of $L$ is to get an element $x$
    via the following two steps:
    \begin{enumerate}
        \item
        With probability ${m_i\over M}$,  select a set $A_i$ among
        $A_1,\cdots, A_{\setCount}$.
        \item
        Get an element $x$ from set $A_i$ via \randomElm$(A_i)$.
    \end{enumerate}
\end{definition}
\newcommand{\fFive}{f_1}      
\newcommand{\fSix}{f_2}       
\newcommand{\fFour}{f_3}      
\newcommand{\fOne}{f_4}       
\newcommand{\fSeven}{f_5}     
\newcommand{\fTwo}{f_6}       

We give some definitions about the parameters and functions that
affect our algorithm below. We assume that $\epsilon\in (0,1)$ is
used to control the accuracy of approximation, and $\gamma\in (0,1)$
is used to control the failure probability. Both parameters are from
the input. In the following algorithm, the two integer parameters
$z_{min}$ and $z_{max}$ with $1\le z_{\min}\le \Thick(L)\le
\ThickMax(L)\le z_{\max}\le \setCount$ can help speed up the
computation. The algorithm is still correct if we use default case
with $z_{min}=1$ and $z_{max}=\setCount$.

\begin{enumerate}
    \item The following parameters are used to control the accuracy
    of approximation at different stages of algorithm:
    \begin{eqnarray}
    \epsilon_0&=&{\epsilon\over 9},  \epsilon_1={\epsilon_0\over 6(\log {\setCount})}, \epsilon_2={\epsilon_1\over 4}, \epsilon_3={\epsilon_0\over 3},\label{epsilons-eqn}\\
    \delta&=&{\epsilon_2\over 2}\label{delta-eqn}.
    \end{eqnarray}
    \item
    The following parameters are used to control the failure probability
    at several stages of the algorithm:
    \begin{eqnarray}
    \gamma_1&=&{\gamma\over 3}, 
    \gamma_2={\gamma\over 6\log {\setCount}}.
    \end{eqnarray}
    \item
    Function $\fFive(.)$ is used to control the number of rounds of the
    algorithm. Its growth rate is mainly determined by the parameter
    $c_1$ that will be determined later:
    \begin{eqnarray}
    \fFive(\setCount)&=&8\setCount^{c_1}\ {\rm with\ }c_1\ge
    0,\label{f5-def-eqn}
    \end{eqnarray}

\item
   Function $\fSix(.)$ is used to check  the number
   of random samples  in $A_j'$ of Stage $1$ in the algorithm.
   We will use different ways to control the accuracy of approximation between the
   case $|A_i|\le {|A_1\cup A_2\cup\cdots\cup A_{\setCount}|\over \fSix(\setCount)}$ and
    the other case $|A_i|> {|A_1\cup A_2\cup\cdots\cup A_{\setCount}|\over \fSix(\setCount)}$. It is mainly used in the proof of
   Lemma~\ref{induction-lemma} that shows it keeps the accuracy of
   approximation when algorithm goes from Stage $i$ to Stage $i+1$.

    \begin{eqnarray}
    \fSix(\setCount)&=&{2v(\delta,z_{\min},z_{\max}, L)\over \epsilon_3} + {2\log {{\setCount}\over  z_{\min}}\over \epsilon_3\log (1+\delta)}. \label{f-6-def} 
    \end{eqnarray}

\item
   Function $\fFour(.)$ is used as a threshold to count the number $t_{i,j}$
   of random samples  in $R_i\cap A_j'$ of Stage $i$ in the algorithm.
   We will use different ways to control the accuracy of approximation between the
   case $t_{i,j}\le \fFour(\setCount)$ and the other case $t_{i,j}>
   \fFour(\setCount)$. It is mainly used in the proof of
   Lemma~\ref{first-phase-lemma} that shows that the number of random samples at Stage $1$ will provide enough accuracy of
   approximation.

    \begin{eqnarray}
    \fFour(\setCount)&=&{\fFive(\setCount)}\cdot {6\ln
        {2\over
            \gamma_2}\over \epsilon_2^2 }. \label{f4-def-eqn}
    \end{eqnarray}

\item
    Function $\fOne(.)$ is used to determine the growth rate of
    function Function $\fSeven(.)$, which is defined by equation
    (\ref{fSeven-def-eqn}).

    \begin{eqnarray}
    \fOne(\setCount)&=& {\fSix(\setCount)\log {{\setCount}^2\over
            \epsilon_1}\over 6\epsilon_1^2 }+{\fFour(\setCount)\over
        \epsilon_2\fFive(\setCount)}.\label{f-1-def}
    \end{eqnarray}

    \item Function $\fSeven(.)$ determines the number of random samples from
    the input sets in the beginning of the algorithm:
    \begin{eqnarray}
    \fSeven(\setCount)&=&{m\fOne(\setCount)\over
    z_{\min}}.\label{fSeven-def-eqn}
    \end{eqnarray}
    \item The following parameter is also used to control failure
    probability in a stage of the algorithm:
    \begin{eqnarray}
    \gamma_3&=&{\gamma_2\over 2\fSeven(\setCount)}.
    \end{eqnarray}
    \item
    Function $\fTwo(.)$ affects the number of random indices in the
    range $\{1,2,\cdots, \setCount\}$. Those random indices will be used
    to choose input sets to detect the approximate $T(x,L)$ for those
    random samples $x$:

    \begin{eqnarray}
    \fTwo(\setCount)&=&\fFive(\setCount)\left({24\over \epsilon_1^2}\cdot {\ln {2\over \gamma_3}}\right). \label{f2-def-eqn}
    \end{eqnarray}
\end{enumerate}
\begin{algorithm}[H]
    \caption{ApproximateUnion($L, z_{\min}, z_{\max}, M, \gamma,
        \epsilon$)}\label{euclid}
    $\mathbf{Input:}$ $L$ is a list of $m$ sets $A_1, A_2,\cdots,
    A_{m}$ with $m\ge 2$, $m_i\in [(1-\beta_L)|A_i|,
    (1+\beta_R)|A_i|]$ and $(\alpha_L, \alpha_R)$-biased random
    generator RandomElement($A_i$) for $i=1,2,\cdots, m$, integers
    $z_{min}$ and $z_{max}$ with $1\le z_{\min}\leq minThickness(L)\leq
    maxThickness(L)\le z_{\max}\le m$, parameter $\gamma\in
    (0,1)$ to control the failure probability, parameter $\epsilon\in
    (0,1)$ to control the accuracy of approximation, and
    $M=m_1+m_2+\cdots+m_{m}$ as the sum of sizes of input sets.\\
    $\mathbf{Output:}$ $sum\cdot M.$
    \begin{algorithmic}[1]
        \State\qquad Let $h_1=f_5(m)$
        \State\qquad Let $i=1$
        \State\qquad Let currentThickness$_1=z_{max}$ \label{ct1-line}
        \State\qquad Let $s_1={m\over{currentThickness_1}}$
        \State\qquad Let $s_1'=1$
        \State\qquad Let $sum=0$
        \State\qquad Obtain a set $R_1$ of $h_1$ random choices of $L$ (see
        Definition~\ref{random-choice-def})\label{R1-select-line}
        \State\qquad Stage $i$
        \State\qquad\qquad Let $u_i=s_i\cdot f_6(m)$
        \State\qquad\qquad Select $u_i$ random indices $H_i=\{k_1,\cdots,
        k_{u_i}\}$ from $\{1,2,\cdots, m\}$
        \State\qquad\qquad Compute $S(x, H_i)$ for each $x\in R_i$
        \State\qquad\qquad Let $V_i$ be the subset of $R_i$ with elements $x$
        satisfying $S(x, H_i)\geq \frac{currentThickness_i}{2f_1(m)\cdot m}\cdot u_i$\label{alg-cond}
        \State\qquad\qquad Let $sum=sum+s_i'\sum\limits\limits_{x\in V_i} {u_i\over S(x, H_i)m}$
        \State\qquad\qquad Let $currentThickness_{i+1}={currentThickness_i\over
            f_1(m)}$\label{cti-line}
        \State\qquad\qquad Let $s_{i+1}={m\over currentThickness_{i+1}}$
        \State\qquad\qquad Let $h_{i+1}={h_1\over  s_{i+1}}$\label{hi-line}
        \State\qquad\qquad If $(|R_i|-|V_i|< h_{i+1})$
        \State\qquad\qquad Then
        \State\qquad\qquad $\{\label{if-else-beginning-line}$
        \State\qquad\qquad\qquad Let $R_{i+1}=R_i-V_i$
        \State\qquad\qquad\qquad Let $a_i=1$
        \State\qquad\qquad \}\label{if-else-middle-line}
        \State\qquad\qquad Else
        \State\qquad\qquad \{
        \State\qquad\qquad\qquad Let $R_{i+1}$ be a set of random $h_{i+1}$
        samples from $R_i-V_i$
        \State\qquad\qquad\qquad Let $a_i= {|R_i|-|V_i|\over h_{i+1}}$
        \State\qquad\qquad \}\label{if-else-ending-line}
        \State\qquad\qquad Let $s_{i+1}'=s_i'\cdot a_i$
        \State\qquad\qquad Let $i=i+1$
        \State\qquad\qquad If  ($currentThickness_i<z_{min}$) \label{until-line}
        \State\qquad\qquad\qquad Return $sum\cdot M$ and terminate the algorithm
        \State\qquad\qquad Else
        \State\qquad\qquad\qquad Enter the next Stage $i$
    \end{algorithmic}
\end{algorithm}

We let $M=m_1+m_2+\cdots+m_{\setCount}$ and $z_{\min}$ be part of
the input of the algorithm. It makes the algorithm be possible to
run in a sublinear time when $z_{\min}\ge \setCount^a$ for a fixed
$a>0$. Otherwise, the algorithm has to spend $\Omega(\setCount)$
time to compute $M$.

\subsection{Proof of Algorithm Performance}

The accuracy and complexity of algorithm ApproximateUnion(.) will be proven in the following lemmas.
Lemma~\ref{basic-facts-lemma} gives some basic properties of the
algorithm.  Lemma~\ref{first-phase-lemma} shows that $R_1$ has
random samples are used so that $F(R_1, h_1,
1)\left(\sum\limits_{i=1}^{\setCount}m_i\right)$ is an accurate approximation for
$\sum\limits_{i=1}^{\setCount}\sum\limits_{x\in A_i} {1\over T(x,L)}$.

\begin{lemma}\label{basic-facts-lemma} The algorithm ApproximateUnion(.) has the following
    properties:
    \begin{enumerate}
        \item\label{f1-ineqn}
        $g^*(\epsilon_1)^{\fOne(\setCount)\over \fSix(\setCount)}\le
        {\epsilon_1\over \setCount^2}$.
        \item\label{v-bound-line}
        $v(\delta,z_{\min}, z_{\max},L)=\bigO\left({\log {z_{\max}\over
                z_{\min}}\over \delta}\right)$.
        \item\label{f-6-def-line}
        ${2v(\delta,z_{\min}, z_{\max},L)\over \fSix(\setCount)}\le
        \epsilon_3$ and $\fSix(\setCount)=\bigO\left({\log{\setCount\over
                z_{\min}}\over \delta\epsilon_3}\right)$.
        \item\label{Ri-def-line}
        $R_i$ contains at most $h_i$ items.
        \item\label{f-3-big-line}
        $g^*(\epsilon_2)^{\fFour(\setCount)\over \fFive(\setCount)}\le {\gamma_2\over 2}$.
        \item\label{f-1-line}
        $\fOne(\setCount)=\bigO\left({1\over \epsilon^4}\left(\log{\setCount\over
        z_{\min}}\right)\cdot \log {{\setCount}\over \epsilon}\cdot
    \left(\log {\setCount}\right)^3+{(\log \setCount)^3(\ln {1\over
            \gamma}+\log\setCount)\over \epsilon^3}\right)$.

        \item\label{f-6-line}
         $\fTwo(\setCount)=\bigO\left({\setCount^{c_1}\over \epsilon^2}\cdot (\log
\setCount)^2(\log {2\over \gamma}+\log \setCount+\log
\fOne(\setCount)) \right)$.
    \end{enumerate}
\end{lemma}

\begin{proof} The statements are easily proven according to the
    setting in the algorithm.

    Statement ~\ref{f1-ineqn}: It follows from equations~(\ref{g*(x)-eqn})
    and (\ref{f-1-def}).

    Statement ~\ref{v-bound-line}: By
    Definition~\ref{partition-thickness-defn}, we need
    $v(\delta,z_{\min}, z_{\max},L)$ with
    $z_{min}(1+\delta)^{v(\delta,z_{\min}, z_{\max},L)}\ge z_{max}$.
    Thus, we have $v(\delta,z_{\min}, z_{\max},L)\le 2\left({\log
        {z_{\max}\over z_{\min}}\over \log
        (1+\delta)}\right)=\bigO\left({\log {z_{\max}\over z_{\min}}\over
        \delta}\right)$ since $\log (1+\delta)=\Theta(\delta)$.

    Statement ~\ref{f-6-def-line}: It is easy to see that $\log
    (1+\delta)=\Theta(\delta)$ and $1\le z_{max}\le \setCount$. It
    follows from equation (\ref{f-6-def}), and
    Statement ~\ref{v-bound-line}.

    Statement ~\ref{Ri-def-line}: It follows from
    lines~\ref{if-else-beginning-line} to~\ref{if-else-ending-line} in
    the algorithm.

    Statement~\ref{f-3-big-line}: It follows from equation (\ref{f4-def-eqn}).

    Statement ~\ref{f-1-line}:
        By equation
    (\ref{f-1-def}), Statement \ref{f-6-def-line} and equation
    (\ref{epsilons-eqn}), we have
    \begin{eqnarray*}
    \fOne(\setCount)&=& {\fSix(\setCount)\log {{\setCount}^2\over
            \epsilon_1}\over 6\epsilon_1^2 }+{\fFour(\setCount)\over
        \epsilon_2\fFive(\setCount)}\\
    &\le&\left(6\fSix(\setCount)\cdot \log {{\setCount}^2\over
        \epsilon_1}\cdot \left({\log
        {\setCount}\over\epsilon_0}\right)^2+{\fFour(\setCount)\over
        \epsilon_2\fFive(\setCount)}\right)\\
    &=&\bigO\left(\fSix(\setCount)\cdot \log {{\setCount}^2\over
        \epsilon_1^2}\cdot \left({\log
        {\setCount}\over\epsilon_0}\right)^2+{1\over \epsilon_2}\cdot
        {\ln
        {2\over \gamma_2}\over \epsilon_2^2 }\right)\\
    &=&\bigO\left({1\over \delta\epsilon_3\epsilon_0^2}(\log{\setCount\over
        z_{\min}})\cdot \log {{\setCount}\over \epsilon_1}\cdot (\log
    {\setCount})^2+{\log \setCount\over \epsilon_0}\cdot {\ln
        {2\over \gamma_2}\over \epsilon_2^2}\right)\\
    &=&\bigO\left({1\over \delta\epsilon_3\epsilon_0^2}(\log{\setCount\over
        z_{\min}})\cdot \log {{\setCount}\over \epsilon_1}\cdot (\log
    {\setCount})^2+{\log \setCount\over \epsilon_0}\cdot {\ln
        {2\over \gamma_2}\over \epsilon_2^2}\right)\\
    &=&\bigO\left({1\over \epsilon^4}\left(\log{\setCount\over
        z_{\min}}\right)\cdot \log {{\setCount}\over \epsilon}\cdot
    \left(\log {\setCount}\right)^3+{(\log \setCount)^3(\ln {1\over
            \gamma}+\log\setCount)\over \epsilon^3}\right).
    \end{eqnarray*}

    Statement~\ref{f-6-line}: By equation (\ref{f2-def-eqn}), we
    have
        \begin{eqnarray}
    \fTwo(\setCount)&=&\fFive(\setCount)\left({24\over \epsilon_1^2}\cdot {\ln {2\over \gamma_3}}\right)\\
        &=&\bigO\left({\setCount^{c_1}\over \epsilon_1^2}\cdot {\ln {2\over \gamma_3}}\right)\\
         &=&\bigO\left({\setCount^{c_1}\over \epsilon^2}\cdot (\log \setCount)^2(\log {2\over \gamma_2}+\ln \fSeven(\setCount))
            \right)\\
         &=&\bigO\left({\setCount^{c_1}\over \epsilon^2}\cdot (\log \setCount)^2(\log {2\over \gamma_2}+\ln \setCount+\ln\fOne(\setCount))
            \right)\\
         &=&\bigO\left({\setCount^{c_1}\over \epsilon^2}\cdot (\log \setCount)^2(\log
{2\over \gamma_2}+\log \setCount+\log \fOne(\setCount)) \right)\\
&=&\bigO\left({\setCount^{c_1}\over \epsilon^2}\cdot (\log
\setCount)^2(\log {2\over \gamma}+\log \setCount+\log
\fOne(\setCount)) \right).
    \end{eqnarray}
\end{proof}

Lemma~\ref{stage-lemma} gives an upper bound for the number of
rounds for the algorithm. It shows how round complexity depends on
$z_{max}, z_{min}$ and $\fFive(.)$.

\begin{lemma}\label{stage-lemma}
    The number of rounds of the algorithm is $\bigO\left({\log
        {z_{max}\over z_{min}}\over \log \fFive(\setCount)}\right)$.
\end{lemma}

\begin{proof} By line~\ref{ct1-line} of the algorithm, we have $\currentT_1=z_{max}$. Variable\\
    $\currentT_i$ is reduced by a factor $\fFive(\setCount)$
    each phase as $\currentT_{i+1}={\currentT_i\over \fFive(\setCount)}$
    by line~\ref{cti-line} of the algorithm. By the termination
    condition of line~\ref{until-line} of the algorithm,  if $y$ is the
    number of phases of the algorithm, we have $y\le y'$, where $y'$ is
    any integer with $ {z_{max}\over \fFive(\setCount)^{y'}}<z_{min}$.
    Thus, $y= \bigO\left({\log {z_{max}\over z_{min}}\over \log
        \fFive(\setCount)}\right)$.
\end{proof}

Lemma~\ref{first-phase-lemma} shows the random samples, which are
saved in $R_1$ in the beginning of the algorithm, will be enough to
approximate the size of set union via $F(R_1, h_1, 1)M$. In the next
a few rounds, algorithm will approximate $F(R_1, h_1, 1)$.

\begin{lemma}\label{first-phase-lemma}
    With probability at least $1-\gamma_1$, $F(R_1, h_1, 1)M\in
    [\frac{(1-\epsilon_0){(1-\alpha_L)(1-\beta_L)}}{1+\delta}|A_1\cup\cdots \cup
    A_{\setCount}|,\,
    (1+\epsilon_0){(1+\alpha_R)(1+\beta_R)(1+\delta)}|A_1\cup\cdots \cup
    A_{\setCount}|]$.
\end{lemma}

\begin{proof}
    Let $A=|A_1\cup\cdots\cup A_{\setCount}|$ and
    $U=|A_1|+|A_2|+\cdots+|A_{\setCount}|$. For an arbitrary set $A_i$
    in the list $L$, and an arbitrary element $x\in A_i$, with at least
    the following probability $x$ is selected via $\randomElm(A_i)$ at
    line~\ref{R1-select-line} of Algorithm ApproximateUnion(.),
    \begin{eqnarray*}
        {m_i\over m_1+m_2+\cdots+m_{\setCount}}{1-\alpha_L\over |A_i|}&\ge&
        {(1-\beta_L)|A_i|\over M}{1-\alpha_L\over |A_i|}\\
        &=& {(1-\beta_L)(1-\alpha_L)\over M}.
    \end{eqnarray*}

    Similarly, with at most the following probability $x$ is chosen via
    $\randomElm(A_i)$ at line~\ref{R1-select-line} of Algorithm
    ApproximateUnion(.),
    \begin{eqnarray*}
        {m_i\over m_1+m_2+\cdots+m_{\setCount}}{1+\alpha_R\over |A_i|}&\le&
        {(1+\beta_R)|A_i|\over M}{1+\alpha_R\over |A_i|}\\
        &=& {(1+\beta_R)(1+\alpha_R)\over M}.
    \end{eqnarray*}

    Define $\rho_L=1-{(1-\beta_L)(1-\alpha_L)}$ and
    $\rho_R={(1+\beta_R)(1+\alpha_R)}-1$. Each element $x$ in $A_1\cup
    A_2\cup\cdots \cup A_{\setCount}$ is selected with probability in
    $\left[{(1-\rho_L)T(x,L)\over M}, {(1+\rho_R)T(x,L)\over M}\right]$.

    Define $T_1=\left\{A_j': |A_j'|\le {A\over
        \fSix(\setCount)}\right\}$, and $T_2=\left\{A_j': |A_j'|> {A\over
        \fSix(\setCount)}\right\}$ (see ~\ref{partition-def} of Definition~\ref{partition-thickness-defn}). Let $t_j=\min\left\{T(x,
    L): x\in A_j'\right\}$.
    We discuss two cases:

    Case 1: $A_j'\in T_1$.
    When one element $x$ is chosen, the probability that
    $x\in A_j'$ is in the range $\left[{(1-\rho_L)t_j|A_j'|\over
        M},{(1+\rho_R)(1+\delta)t_j|A_j'|\over M}\right]$. Let
    $p_j={(1+\rho_R)(1+\delta)t_j\cdot {A\over \fSix(\setCount)}\over
        M}$. Since $z_{\min}\le \Thick(L)$, we have $z_{\min}\le
    \Thick(L)\le t_j$. It is easy to see that $mA\ge U$. We have
    \begin{eqnarray}
    p_jh_1&=&{(1+\rho_R)(1+\delta)t_j\cdot {A\over \fSix(\setCount)}\over M}\cdot {m\fOne(\setCount)\over z_{\min}}\nonumber\\
    &\ge& {(1+\rho_R)(1+\delta)\cdot \fOne(\setCount)m A\over
        \fSix(\setCount) M}\nonumber\\
    &\ge& {(1+\rho_R)(1+\delta)\cdot \fOne(\setCount)U\over
        \fSix(\setCount) M}\nonumber\\
    &\ge& {(1+\rho_R)(1+\delta)\cdot \fOne(\setCount)\over
        \fSix(\setCount)(1+\beta_R) }\nonumber\\
    &=& {(1+\alpha_R)(1+\delta)\cdot \fOne(\setCount)\over
        \fSix(\setCount) }.\label{case1-a-ineqn}
    \end{eqnarray}
    Let
    $\omega_1(\setCount)={(1+\alpha_R)(1+\delta)\cdot
        \fOne(\setCount)\over \fSix(\setCount) }$. Thus, $p_jh_1\ge
    \omega_1(\setCount)$.

    Let $R_{1,j}$ be the elements of $R_1$ and also in $A_j'$. By
    Theorem~\ref{ourchernoff2-theorem}, with probability at most
    $P_j=g^*(1)^{p_j\cdot h_1}\le g^*(1)^{\omega_1(\setCount)}\le
    {\gamma_1\over 2}$ (by equation (\ref{f-6-def}), equation
    (\ref{f-1-def}) and inequality (\ref{case1-a-ineqn})), there are
    more than $2p_jh_1={2(1+\rho_R)(1+\delta)t_j\cdot {A\over
            \fSix(\setCount)}\over M}\cdot h_1$ elements to be chosen from
    $A_j'$ into $R_1$. Thus,
    \begin{eqnarray}
    F'(R_{1,j})\le {2p_jh_1\over t_j}= {2(1+\rho_R)(1+\delta)\cdot
        A\over \fSix(\setCount) M}\cdot h_1,\label{case1-ineqn}
    \end{eqnarray}
    with probability at
    most $P_j$ to fail.

    Case 2: $A_j'\in T_2$. When $h_1$ elements are selected to $R_1$,
    let $v_j$ be the number of elements selected in $A_j'$. When one
    element $x$ is chosen, the probability that $x\in A_j'$ is in the
    range $\left[{(1-\rho_L)t_j|A_j'|\over
        M},{(1+\rho_R)(1+\delta)t_j|A_j'|\over M}\right]$.

    Let $p_{j,1}={(1-\rho_L)t_j|A_j'|\over M}$ and
    $p_{j,2}={(1+\rho_R)(1+\delta)t_j|A_j'|\over M}$.

    We have
    \begin{eqnarray}
    p_{j,1}h_1&=&{(1-\rho_L)t_j|A_j'|\over M}\cdot h_1\nonumber\\
    &>& {(1-\rho_L)t_j{A\over \fSix(\setCount)}\over M}\cdot h_1\ge
    {(1-\rho_L)z_{min}A\cdot h_1\over \fSix(\setCount) M}=
    {(1-\rho_L)z_{min}A\cdot m\fOne(\setCount)\over z_{min}\fSix(\setCount) M}\nonumber\\
    &\ge& {(1-\rho_L)\fOne(\setCount)\over \fSix(\setCount)(1+\beta_R)
    }.\label{case2-a-ineqn}
    \end{eqnarray}

    We have
    \begin{eqnarray}
    p_{j,2}h_1&=&{(1+\rho_R)(1+\delta)t_j|A_j'|\over M}\cdot h_1\nonumber\\
    &>& {(1+\rho_R)(1+\delta)t_j\cdot {A\over \fSix(\setCount)}\over
        M}\cdot h_1\nonumber\\
    &\ge& {(1+\rho_R)(1+\delta)z_{min}\cdot A\over \fSix(\setCount)
        M}\cdot
    h_1\nonumber\\
    &=& {(1+\rho_R)(1+\delta)z_{min}\cdot A\over \fSix(\setCount)
        M}\cdot
    {m\fOne(\setCount)\over z_{min}}\nonumber\\
    &\ge& {(1+\rho_R)(1+\delta)\cdot \fOne(\setCount)\over
        \fSix(\setCount)(1+\beta_R) }\nonumber\\
    &\ge& {(1+\alpha_R)(1+\delta)\cdot \fOne(\setCount)\over
        \fSix(\setCount) }.\label{case2-b-ineqn}
    \end{eqnarray}
    With probability at
    most $g^*(\epsilon_3)^{p_{j,1}\cdot h_1}\le {\gamma_1\over 4}$ (by
    equation (\ref{f-6-def}), equation (\ref{f-1-def}) and inequality
    (\ref{case2-a-ineqn})),
    \begin{eqnarray*}
        v_j<
        {(1-\epsilon_3)(1-\rho_L)t_j|A_j'|\over M}\cdot
        h_1=(1-\epsilon_3)(1-\rho_L)t_jh_1\cdot {|A_j'|\over M}.
    \end{eqnarray*}

    With probability at most $g^*(\epsilon_3)^{p_{j,2}\cdot h_1}\le
    {\gamma_1\over 4}$ (by equation (\ref{f-6-def}), equation
    (\ref{f-1-def}) and inequality (\ref{case2-b-ineqn})),
    \begin{eqnarray*}
        v_j>
        {(1+\epsilon_3)(1+\rho_R)(1+\delta)t_j|A_j'|\over M}\cdot
        h_1=(1+\epsilon_3)(1+\rho_R)(1+\delta)t_jh_1\cdot {|A_j'|\over M}.
    \end{eqnarray*}

    Therefore, with probability at least $1-\gamma_1/2$, we have
    \begin{eqnarray}
    v_j\in \left[(1-\epsilon_3)(1-\rho_L)t_jh_1\cdot {|A_j'|\over M},
    (1+\epsilon_3)(1+\rho_R)(1+\delta)t_jh_1\cdot {|A_j'|\over
        M}\right].\label{case2-ineqn}
    \end{eqnarray}
    Thus, we have that there are sufficient elements of $A_j'$ to be
    selected with high probability, which follows from
    Theorem~\ref{chernoff1-theorem} and
    Theorem~\ref{ourchernoff2-theorem}.

    In the rest of the proof, we assume that inequality
    (\ref{case1-ineqn}) holds if the condition of Case 1 holds, and
    inequality (\ref{case2-ineqn}) holds if the condition of Case 2
    holds.

    Now we consider
    \begin{eqnarray}
    &&F(R_1, h_1, 1)={1\over h_1}\sum\limits_{x\in R_1}{1\over T(x, L)}\nonumber\\
    &=&{1\over h_1}\left(\sum\limits_{R_{1,j}\ with \ A_j'\in T_1}\sum\limits_{x\in R_{1,j}}{1\over T(x, L)}+\sum\limits_{R_{1,j}\ with \ A_j'\in T_2}\sum\limits_{x\in R_{1,j}}{1\over T(x, L)}\right)\nonumber\\
    &\le&{1\over h_1}{2(1+\rho_R)(1+\delta)v(\delta,z_{\min},
        z_{\max}, L)\cdot A\over \fSix(m) M}\cdot h_1\nonumber\\
    &&+\frac{1}{h_1}\sum\limits_{R_{1,j}\
        with\ A_j'\in T_2}\sum\limits_{x\in
        R_{1,j}}{1\over T(x, L)}\nonumber\\
    &\leq&{1\over h_1}\left({2(1+\rho_R)(1+\delta)v(\delta,z_{\min},
        z_{\max}, L)\cdot A\over
        \fSix(m) M}\cdot h_1+\sum\limits_{R_{1,j}\ with\ A_j'\in T_2}{v_j\over t_j}\right)\nonumber\\
    &\le&{1\over h_1}\left({2(1+\rho_R)(1+\delta)v(\delta,z_{\min},
        z_{\max}, L)\cdot A\over \fSix(m) M}\cdot h_1\right) \nonumber\\
    &&+{1\over h_1}\left(\sum\limits_{R_{1,j}\ with\ A_j'\in
        T_2}{(1+\epsilon_3)(1+\rho_R)(1+\delta)h_1\cdot {|A_j'|\over M}}\right)\nonumber\\
    &=&{2(1+\rho_R)(1+\delta)v(\delta,z_{\min}, z_{\max}, L)\cdot A\over
        \fSix(m) M}\nonumber\\
    &&+\sum\limits_{R_{1,j}\ with \ A_j'\in
        T_2}{(1+\epsilon_3)(1+\rho_R)(1+\delta)\cdot {|A_j'|\over M}}\nonumber\\
    &\le&\left({2(1+\rho_R)(1+\delta)v(\delta,z_{\min}, z_{\max},
        L)\over \fSix(m)
    }+(1+\epsilon_3)(1+\rho_R)(1+\delta)\right){A\over M}\nonumber\\
    &=&\left(1+\epsilon_3+{2v(\delta,z_{\min}, z_{\max}, L)\over
        \fSix(m)
    }\right)(1+\rho_R)(1+\delta){A\over M}\label{epsilon5-ineqn1}\\
    &\le&(1+2\epsilon_3)(1+\rho_R)(1+\delta){A\over M}\label{epsilon5-ineqn2}\\
    &\le&(1+\epsilon_0)(1+\rho_R)(1+\delta){A\over M}\nonumber.
    \end{eqnarray}

    The transition from (\ref{epsilon5-ineqn1}) to
    (\ref{epsilon5-ineqn2}) is by Statement~\ref{f-6-def-line} of
    Lemma~\ref{basic-facts-lemma}. For the lower bound part, we have the
    following inequalities:

    \begin{eqnarray}
    &&F(R_1, h_1, 1) ={1\over h_1}\sum\limits_{x\in R_1}{1\over T(x, L)}\nonumber\\
    &\ge&{1\over h_1}\left(\sum\limits_{R_{1,j}\ with\ A_j'\in T_2}\sum\limits_{x\in
        R_{1,j}}{1\over T(x, L)}\right)\nonumber\\
    &\geq&{1\over h_1}\left(\sum\limits_{R_{1,j}\ with\ A_j'\in T_2}{v_j\over (1+\delta)t_j}\right)\nonumber\\
    &\ge&{1\over h_1}\left(\sum\limits_{R_{1,j}\ with\ A_j'\in
        T_2}{(1-\epsilon_3)(1-\rho_L)h_1\cdot {|A_j'|\over (1+\delta)M}}\right)\nonumber\\
    &=&{(1-\epsilon_3)(1-\rho_L)\over (1+\delta)M}\left(\sum\limits_{R_{1,j}\ with\
        A_j'\in
        T_2}|A_j'|\right)\nonumber\\
    &=&{(1-\epsilon_3)(1-\rho_L)\over (1+\delta)M}\sum\limits_{R_{1,j}\ with\
        A_j'\in T_1}|A_j'|+{(1-\epsilon_3)(1-\rho_L)\over (1+\delta)M}\sum\limits_{R_{1,j}\ with\ A_j'\in
        T_2}|A_j'|\nonumber\\
    &&-{(1-\epsilon_3)(1-\rho_L)\over (1+\delta)M}\sum\limits_{R_{1,j}\ with\ A_j'\in
        T_1}|A_j'|\nonumber\\
    &=&{(1-\epsilon_3)(1-\rho_L)\over (1+\delta)M}\left(A-\sum\limits_{R_{1,j}\ with\
        A_j'\in
        T_1}|A_j'|\right)\nonumber\\
    &\geq&{(1-\epsilon_3)(1-\rho_L)\over (1+\delta)M}\left(A- {v(\delta,z_{\min}, z_{\max}, L)A\over \fSix(m)}\right)\nonumber\\
    &=&\left(1- {v(\delta,z_{\min}, z_{\max}, L)\over
        \fSix(m)}\right)\frac{(1-\epsilon_3)(1-\rho_L)}{(1+\delta)}{A\over M}\label{epsilon5-ineqn3}\\
    &\ge&(1-\epsilon_3)\frac{(1-\epsilon_3)(1-\rho_L)}{(1+\delta)}{A\over M}\label{epsilon5-ineqn4}\\
    &\ge&\frac{(1-\epsilon_0)(1-\rho_L)}{(1+\delta)}{A\over M}\nonumber.
    \end{eqnarray}

    The transition from (\ref{epsilon5-ineqn3}) to
    (\ref{epsilon5-ineqn4}) is by Statement $3$ of
    Lemma~\ref{basic-facts-lemma}. Therefore, $F(R_1, h_1, 1)M\in
    \left[\frac{(1-\epsilon_0)(1-\rho_L)}{(1+\delta)}A,\,(1+\epsilon_0)(1+\rho_R)(1+\delta)A\right]$.
\end{proof}

Lemma~\ref{middle-lemma} shows that at stage $i$, it can approximate
$T(x,L)$ for all random samples with highest $T(x,L)$  in $R_i$.
Those random elements with highest $T(x,L)$ will be removed in stage
$i$ so that the algorithm will look for random elements with smaller
$T(x,L)$ in the coming stages.

\begin{lemma}\label{middle-lemma} After the execution of Stage
    $i$, with probability at least $1-\gamma_2$, we have the following
    three statements:
    \begin{enumerate}
        \item\label{statm1}
        Every  element $x\in R_i$
        with $T(x, L)\ge {\currentT_i\over 4\fFive(\setCount)}$ has $S(x,
        H_i)\in$\\
        $ \left[(1-\epsilon_1){T(x, L)\over {\setCount}}u_i,
        (1+\epsilon_1){T(x, L)\over {\setCount}}u_i\right]$.
        \item\label{statm2}
        Every element $x\in V_i$ with $T(x, L)\ge {\currentT_i\over
            \fFive(\setCount)}$, it satisfies the condition in
        line~\ref{alg-cond} of the algorithm.

        \item\label{statm3}
        Every element $x\in V_i$ with $T(x, L)< {\currentT_i\over
            4\fFive(\setCount)}$, it does not satisfy the condition in
        line~\ref{alg-cond} of the algorithm.
    \end{enumerate}
\end{lemma}

\begin{proof}
    It follows from Theorem~\ref{chernoff1-theorem} and
    Theorem~\ref{ourchernoff2-theorem}. There are
    $u_i=s_i\fTwo(\setCount)$ indices are selected among $\{1,2,\cdots,
    \setCount\}$. Let $p={T(x, L)\over {\setCount}}$.

    Statment~\ref{statm1}: We have $pu_i= {T(x, L)\over
        {\setCount}}\cdot s_i\fTwo(\setCount)\ge {\currentT_i\over
        4\fFive(\setCount)}\cdot {1\over {\setCount}}\cdot
    s_i\fTwo(\setCount)={\currentT_i\over 4\fFive(\setCount)}\cdot
    {1\over {\setCount}}\cdot {\setCount\over\currentT_i}\cdot
    \fTwo(\setCount)={\fTwo(\setCount)\over 4\fFive(\setCount)}$.

    With probability at most $P_1=g^*(\epsilon_1)^{pu_i}\le
    {\gamma_3\over 2}$ (by equations (\ref{f4-def-eqn}), and
    (\ref{f2-def-eqn})), $S(x, H_i)<(1-\epsilon_1){T(x, L)\over
        {\setCount}}u_i$. With probability at most
    $P_2=g^*(\epsilon_1)^{pu_i}\le {\gamma_3\over 2}$ (by equations
    (\ref{f4-def-eqn}) and (\ref{f2-def-eqn})), $S(x,
    H_i)>(1+\epsilon_1){T(x, L)\over {\setCount}}u_i$.

    There are at most $h_i$ elements in $R_i$ by
    Statement~\ref{Ri-def-line} of Lemma~\ref{basic-facts-lemma}.
    Therefore, with probability at most $h_i(P_1+P_2)\le h_1(P_1+P_2)\le
    h_1\cdot \gamma_3= {\gamma_2\over 2}$, $$S(x, H_i)\not\in
    \left[(1-\epsilon_1){T(x, L)\over {\setCount}}u_i,
    (1+\epsilon_1){T(x, L)\over {\setCount}}u_i\right].$$

    Statement~\ref{statm2}: This statement of the lemma follows from
    Statement~\ref{statm1}.

    Statement~\ref{statm3}: This part of the lemma follows from
    Theorem~\ref{chernoff1-theorem} and
    Theorem~\ref{ourchernoff2-theorem}. For $x\in V_i$ with $T(x, L)<
    {\currentT_i\over 4\fFive(\setCount)}$, let $p={\currentT_i\over
        4\fFive(\setCount)}$. With probability at most $g^*(1)^{pu_i}\le
    {\gamma_3\over 2}$ (by equations (\ref{f4-def-eqn}), and
    (\ref{f2-def-eqn})), we have $S(x, H_i)\ge 2pu_i$.
    There are at most $h_i$ elements in $R_i$ by
    Statement~\ref{Ri-def-line} of Lemma~\ref{basic-facts-lemma}.
    Therefore, with probability at most $h_1\cdot \gamma_3\le {\gamma_2\over 2}$,  there exists one $x\in V_i$ with $T(x, L)<
    {\currentT_i\over 4\fFive(\setCount)}$ to satisfy $S(x, H_i)\ge
    2pu_i$.
\end{proof}

\begin{lemma}\label{ex-lemma}
    Let $x$ and $y$ be positive real numbers with $1\le y$. Then we have:
    \begin{enumerate}
        \item
        $1-xy<(1-x)^y$.
        \item
        If $xy<1$, then $(1+x)^y<1+2xy$.
        \item\label{sub-add-case}
        If $x_1, x_2\in [0,1)$, then $1-x_1-x_2\le (1-x_1)(1-x_2)$, and
        $(1+x_1)(1+x_2)\le 1+2x_1+x_2$.
    \end{enumerate}
\end{lemma}

\begin{proof}
    By Taylor formula, we have $(1-x)^y=1-xy+{y\cdot (y-1)\over 2}\theta^2$
    for some $\theta\in [0,x]$. Thus, we have $(1-x)^y\ge 1- y x$. Note
    that the function $(1+{1\over z})^z$ is increasing, and
    $\lim_{z\rightarrow +\infty}(1+{1\over z})^z= e$. We also have
    $(1+x)^y\le (1+x)^{{1\over x}\cdot xy}\le e^{xy}\le 1+xy+(xy)^2\le
    1+2xy$.

    It is trivial to verify Statement~\ref{sub-add-case}. $1-x_1-x_2\le
    (1-x_1)(1-x_2)$.
    Clearly, $(1+x_1)(1+x_2)= 1+x_1+x_2+x_1x_2\le 1+2x_1+x_2$.
\end{proof}

Lemma~\ref{induction-lemma} shows that how to gradually approximate
$F(R_1, h_1, 1)M$ via several rounds. It shows that the left random
samples stored in $R_{i+1}$ after stage $i$ is enough to approximate
$F'(R_i-V_i)$.

\begin{lemma}\label{induction-lemma} Let $y$ be the number of stages.
    Let $V_i$ be the set of elements
    removed from $R_i$ in Stage $i$. Then we have the following facts:
    \begin{enumerate}
        \item\label{case1-lemma}
        With probability at least $1-\gamma_2$,
        $a_iF'(R_{i+1})\in [(1-\epsilon_1)F'(R_i-V_i),
        (1+\epsilon_1)F'(R_i-V_i)]$, and

        \item\label{case3-lemma}
        With probability at least $1-2y\gamma_2$, $\sum_{i=1}^{y}
        s_i'F'(V_i)\in [(1-y\epsilon_1)S, (1+2y\epsilon_1)S]$, where
        $S=F(R_1,h_1,1)$.
    \end{enumerate}
\end{lemma}

\begin{proof}Let $h_i'=h_i-|V_i|$.
    If an local is too small, it does not affect the global sum much. In
    $R_{i+1}$, we deal with the elements $x$ of $T(x,L)<
    {\currentT_{i}\over \fFive(\setCount)}$. By
    Lemma~\ref{middle-lemma}, with probability at least $1-\gamma_2$,
    $R_i-V_i$ does not contain any $x$  with $T(x,L)\ge
    {\currentT_{i}\over \fFive(\setCount)}$.

    Let $t_{i,j}$ be the number of elements of $A_j'$ in $R_i$ with
    multiplicity. Let $B_{i,j}$ be the set of elements in both $R_{i}$
    and $A_j'$ with multiplicity.

    Statement~\ref{case1-lemma}:  We discuss two cases:

    Case 1: $|R_i|-|V_i|< h_{i+1}$. This case is trivial  since
    $R_{i+1}=R_i-V_i$ and $a_i=1$ according to the algorithm
    (line~(\ref{if-else-beginning-line}) to line
    (\ref{if-else-middle-line})).

    In the following Case 2, we assume the condition of Case 1 is false.
    Thus, $h_i'\ge h_{i+1}$.

    Case 2: $|R_i|-|V_i|\ge h_{i+1}$. We have
    \begin{eqnarray}
    F'(R_{i}-V_i)&\ge& {h_i'\over {\currentT_i\over
            \fFive(\setCount)}}\ge
    {h_{i+1}\over {\currentT_i\over \fFive(\setCount)}}\nonumber\\
    &=&{{h_1\over s_{i+1}}\cdot \fFive(\setCount)\over
        \currentT_i}={\fOne(\setCount)\over z_{min}}.\label{F'-ineqn}
    \end{eqnarray}

    Two subcases are discussed below.

    Subcase 2.1: $t_{i,j}\le \fFour(\setCount)$, in this case, $B_{i,j}$
    has a small impact for the global sum.

    Let $p={\fFour(\setCount)\over h_i'}$. By
    Theorem~\ref{chernoff1-theorem} and
    Theorem~\ref{ourchernoff2-theorem}, with probability at least
    $1-g^*(1)^{ph_{i+1}}=1-g^*(1)^{\fFour(\setCount)\over
        \fFive(\setCount)}\ge 1-{\gamma_2\over 2}$ (by equation
    (\ref{f4-def-eqn})),
    \begin{eqnarray*}
    |B_{i+1,j}|\le 2ph_{i+1}=2\cdot {\fFour(\setCount)\over h_i'} \cdot
    {h_i\over \fFive(\setCount)}={2\fFour(\setCount)\over
        \fFive(\setCount)}\cdot {h_{i}\over h_i'}\le
    {2\fFour(\setCount)\over \fFive(\setCount)}\cdot {h_{i}\over
        h_{i+1}}\le {2\fFour(\setCount)\over \fFive(\setCount)^2}.
    \end{eqnarray*}

    We
    assume $|B_{i+1,j}|\le {2\fFour(\setCount)\over
        \fFive(\setCount)^2}$. We have $F'(B_{i+1,j})\le {|B_{i+1,j}|\over
        z_{min}}\le {2\fFour(\setCount)\over z_{min}\fFive(\setCount)^2}$.
    Clearly, $a_i\le \fFive(\setCount)$. Thus,
    \begin{eqnarray}
    a_iF'(B_{i+1,j})&\le& \fFive(\setCount)\cdot
    {2\fFour(\setCount)\over
        z_{min}\fFive(\setCount)^2}={2\fFour(\setCount)\over z_{min}\fFive(\setCount)}\nonumber\\
    &=&{2\fFour(\setCount)\over \fOne(\setCount) \fFive(\setCount)}\cdot {\fOne(\setCount)\over z_{min}}\label{ai-Fi-ineqn}\\
    &\le& {2\fFour(\setCount)\over
        \fFive(\setCount)\fOne(\setCount)}\cdot
    F'(R_i-V_i)\label{ai-Fi-ineqn1}\\
    &\le& {2\epsilon_2\cdot F'(R_i-V_i)\over v(\delta,z_{\min},
        z_{\max}, L)}.\label{ai-Fi-ineqn2}
    \end{eqnarray}

    The transition from (\ref{ai-Fi-ineqn}) to (\ref{ai-Fi-ineqn1}) is
    by inequality (\ref{F'-ineqn}). The transition from
    (\ref{ai-Fi-ineqn1}) to (\ref{ai-Fi-ineqn2}) is by inequality
    (\ref{f-1-def}).

    Subcase 2.2: $t_{i,j}> \fFour(\setCount)$ in $R_i$, in this case,
    $B_j'$ does not lose much accuracy.
    From $R_i$ to $R_{i+1}$, $h_{i+1}={h_i\over \fFive(\setCount)}$
    elements are selected.

    Let $q={t_{i,j}\over h_i'}$. We have
    \begin{eqnarray}
    qh_{i+1}={t_{i,j}\over h_i'}\cdot h_{i+1}=t_{i,j}\cdot {h_{i+1}\over
        h_i'}\ge t_{i,j}\cdot {h_{i+1}\over h_i}\ge {\fFour(\setCount)\over
        \fFive(\setCount)}.\label{case2.2-ineqn}
    \end{eqnarray}

    With probability at most $g^*(\epsilon_2)^{qh_{i+1}}\le
    {\gamma_2\over 2}$ (by inequality (\ref{case2.2-ineqn}) and Statement~\ref{f-3-big-line} of Lemma~\ref{basic-facts-lemma}), we have
    that $|B_{i+1,j}|<(1-\epsilon_2)qh_{i+1}$.  With probability at most
    $g^*(\epsilon_2)^{qh_{i+1}}\le {\gamma_2\over 2}$ (by inequality
    (\ref{case2.2-ineqn}) and Statement~\ref{f-3-big-line} of Lemma~\ref{basic-facts-lemma}), we have that
    $|B_{i+1,j}|>(1+\epsilon_2)qh_{i+1}$. They follow from
    Theorem~\ref{chernoff1-theorem} and
    Theorem~\ref{ourchernoff2-theorem}.

    We assume $|B_{i+1,j}|\in \left[(1-\epsilon_2)qh_{i+1},
    (1+\epsilon_2)qh_{i+1}\right]$. Thus, $a_iF'(B_{i+1,j})\in
    [(1-\epsilon_2)t_{i,j}, (1+\epsilon_2)t_{i,j}]$. So,
    $a_iF'(B_{i+1,j})\in \left[{(1-\epsilon_2)F'(R_{i,j})\over
        1+\delta}, (1+\epsilon_2)F'(R_{i,j})(1+\delta)\right]$.

    We have

    \begin{eqnarray}
    a_iF'(R_{i+1})&=&a_i\left(\sum_j F'(B_{i+1, j})\right)\label{to-epsilon4-ineqn0}\\
    &\le& (1+\epsilon_2)(1+\delta)F'(R_{i}-V_i)+{2\epsilon_2F'(R_i-V_i)\over
        v(\delta,z_{\min}, z_{\max}, L)} \cdot v(\delta,z_{\min}, z_{\max}, L)\label{to-epsilon4-ineqn0a}\\
    &\le& ((1+\epsilon_2)(1+\delta)+2\epsilon_2)F'(R_i-V_i)\label{to-epsilon4-ineqn-a2}\\
    &\le& (1+4\epsilon_2)F'(R_i-V_i)\label{to-epsilon4-ineqn-b2}\\
    &\le& (1+\epsilon_1)F'(R_i-V_i)\label{to-epsilon4-ineqn-b3}.
    \end{eqnarray}

    The transition from (\ref{to-epsilon4-ineqn0}) to (\ref{to-epsilon4-ineqn0a}) is
    by inequality (\ref{ai-Fi-ineqn2}).
    The transition from (\ref{to-epsilon4-ineqn-a2})
    to (\ref{to-epsilon4-ineqn-b2}) is based on equation
    (\ref{delta-eqn}). The transition from (\ref{to-epsilon4-ineqn-b2})
    to (\ref{to-epsilon4-ineqn-b3}) is based on equations
    (\ref{epsilons-eqn}).

    We have

    \begin{eqnarray}
    a_iF'(R_{i+1})&=&a_i\left(\sum_j F'(B_{i+1, j})\right)\label{to-epsilon4-ineqn1}\\
    &\ge& {(1-\epsilon_2)F'(R_{i}-V_i)\over
        (1+\delta)}-{2\epsilon_2F'(R_i-V_i)\over
        v(\delta,z_{\min}, z_{\max}, L)} \cdot v(\delta,z_{\min}, z_{\max}, L)\label{to-epsilon4-ineqn1a}\\
    &\ge& \left({(1-\epsilon_2)\over (1+\delta)}-2\epsilon_2\right)F'(R_i-V_i)\label{to-epsilon4-ineqn1b}\\
    &\ge& (1-4\epsilon_2)F'(R_i-V_i)\label{to-epsilon4-ineqn1b3}\\
    &\ge& (1-\epsilon_1)F'(R_i-V_i).\label{to-epsilon4-ineqn1b4}
    \end{eqnarray}

The transition from (\ref{to-epsilon4-ineqn1}) to
    (\ref{to-epsilon4-ineqn1a}) is based on inequality (\ref{ai-Fi-ineqn2}).
     The transition from (\ref{to-epsilon4-ineqn1b3})
    to (\ref{to-epsilon4-ineqn1b4}) is based on equations
    (\ref{epsilons-eqn}).

    Statement~\ref{case3-lemma}: In the rest of the proof, we assume
    that if $|R_i|-|V_i|\ge h_{i+1}$, then $F'(R_{i+1})=F'(R_i-V_i)$,
    and if $|R_i|-|V_i|< h_{i+1}$, then $a_iF'(R_{i+1})\in
    [(1-\epsilon_1)F'(R_i-V_i), (1+\epsilon_1)F'(R_i-V_i)]$.

    In order to prove Statement~\ref{case3-lemma},  we give an inductive
    proof that $s_{k+1}'F'(R_{k+1})+\sum_{i=1}^{k} s_i'F'(V_i)\in
    [(1-\epsilon_1)^kS, (1+\epsilon_1)^kS]$. It is trivial for $k=0$.
    Assume that $s_{k}'F'(R_{k})+\sum_{i=1}^{k-1} s_i'F'(V_i)\in
    [(1-\epsilon_1)^{k-1}S, (1+\epsilon_1)^{k-1}S]$.

    Since $F'(R_{k})=F'(R_k-V_k)+F'(V_k)$, we have
    $a_kF'(R_{k+1})+F'(V_k)\in [(1-\epsilon_1)F'(R_{k}),
    (1+\epsilon_1)F'(R_{k})]$.

    Thus, we have
    \begin{eqnarray*}
        s_{k+1}'F'(R_{k+1})+\sum_{i=1}^{k} s_i'F'(V_i)&=&
        s_{k+1}'F'(R_{k+1})+s_{k}'F'(V_k)+\sum_{i=1}^{k-1} s_iF'(V_i)\\
        &=& s_{k}'(a_kF'(R_{k+1})+F'(V_k))+\sum_{i=1}^{k-1} s_iF'(V_i)\\
        &\le& (1+\epsilon_1)s_{k}'F'(R_{k})+\sum_{i=1}^{k-1} s_iF'(V_i)\\
        &\le& (1+\epsilon_1)\left(s_{k}'F'(R_{k})+\sum_{i=1}^{k-1} s_iF'(V_i)\right)\\
        &\le& (1+\epsilon_1)^kS.
    \end{eqnarray*}

    Similarly,  we have
    \begin{eqnarray*}
        s_{k+1}'F'(R_{k+1})+\sum_{i=1}^{k} s_i'F'(V_i)&=&
        s_{k+1}'F'(R_{k+1})+s_{k}F'(V_k)+\sum_{i=1}^{k-1} s_i'F'(V_i)\\
        &=& s_{k}'\left(a_kF'(R_{k+1})+F'(V_k)\right)+\sum_{i=1}^{k-1} s_i'F'(V_i)\\
        &\ge& (1-\epsilon_1)s_{k}'F'(R_{k})+\sum_{i=1}^{k-1} s_i'F'(V_i)\\
        &\ge& (1-\epsilon_1)\left(s_{k}'F'(R_{k})+\sum_{i=1}^{k-1} s_i'F'(V_i)\right)\\
        &\ge& (1-\epsilon_1)^kS.
    \end{eqnarray*}

    Thus, we have $s_{k+1}'F'(R_{k+1})+\sum_{i=1}^{k} s_i'F'(V_i)\in
    [(1-\epsilon_1)^kS, (1+\epsilon_1)^kS]$.

    Therefore, with probability at least $1-y\gamma_2-y\gamma_2$,
    $\sum_{i=1}^{y} s_i'F'(V_i)\in [(1-\epsilon_1)^yS,
    (1+\epsilon_1)^yS]\subseteq [(1-\epsilon_1y)S, (1+2\epsilon_1y)S]$
    by Lemma~\ref{ex-lemma}.

\end{proof}

Lemma~\ref{time-lemma} gives the time complexity of the algorithm.
The running time depends on several parameters.

\begin{lemma}\label{time-lemma}
    The algorithm ApproximateUnion(.) runs in
    $\bigO\left({{\setCount}\fOne(\setCount)\fTwo(\setCount)\over
        z_{\min}}\cdot \left({\log {z_{\max}\over z_{\min}}\over \log
        \fFive(\setCount)}\right)\right)$ time.
\end{lemma}

\begin{proof}
    Let $y$ be the total number of stages. By Lemma~\ref{stage-lemma},
    we have $y=\bigO\left({\log {z_{max}\over z_{min}}\over \log
        \fFive(\setCount)}\right)$.

    The time of each stage is $t_i=h_i\cdot
    u_i=h_1\fTwo(\setCount)={{\setCount}\over
        z_{\min}}\fOne(\setCount)\fTwo(\setCount)$, which is mainly from
    line~\ref{alg-cond} of the algorithm. Therefore, the total time is
    $\sum\limits_{i=1}^{y}t_i\le {{\setCount}\over z_{\min}}\cdot
    \fOne(\setCount)\fTwo(\setCount)y$.

\end{proof}
We have Theorem~\ref{nonadaptive-thm} to show the performance of the
algorithm. The algorithm is sublinear if $\Thick(L)\ge
{\setCount}^a$ for a fixed $a>0$, and has a $z_{\min}$ with
$\Thick(L)\ge z_{\min}\ge {\setCount}^{b}$ for a positive fixed $b$
($b$ may not be equal to $a$) to be part of input to the algorithm.

\begin{theorem}\label{nonadaptive-thm}
    The algorithm ApproximateUnion(.) takes
    $\bigO\left({{\setCount}\fOne(\setCount)\fTwo(\setCount)\over
        z_{\min}}\cdot \left({\log {z_{\max}\over z_{\min}}\over \log
        \fFive(\setCount)}\right)\right)$ time and $\bigO\left({\log
        {z_{\max}\over z_{\min}}\over \log \fFive(\setCount)}\right)$ rounds
    such that with probability at least $1-\gamma$, it gives a
    $$sum\cdot M\in
    [{(1-\epsilon)(1-\alpha_L)(1-\beta_L)}\cdot |A_1\cup\cdots \cup
    A_{\setCount}|, {(1+\epsilon)(1+\alpha_R)(1+\beta_R)}\cdot
    |A_1\cup\cdots \cup A_{\setCount}|],$$ where $z_{\min}$ and
    $z_{\max}$ are parameters with $1\le z_{\min}\le \Thick(L)\le
    \ThickMax(L)$ \\$\le z_{\max}\le \setCount$, where functions
    $\fFive(.)$,  $\fOne(.)$, and $\fTwo(.)$ are defined in equations~(\ref{f-1-def}), (\ref{f5-def-eqn}), and~(\ref{f2-def-eqn}), respectively.
\end{theorem}

\begin{proof}
    Let $y$ be the number of stages.
    By Lemma~\ref{middle-lemma}, with probability at least
    $1-y\gamma_2$, $$sum\in \left[\left(1-\epsilon_1\right)\sum\limits_{i=1}^{y} s_i'F'(V_i),\, \left(1+\epsilon_1\right)\sum\limits_{i=1}^{y} s_i'F'(V_i)\right].$$

    By Lemma~\ref{induction-lemma}, with probability at least $1-2y\gamma_2$,
    $$\sum_{i=1}^{y}
    s_i'F'(V_i)\in \left[\left(1-y\epsilon_1\right)F(R_1,h_1,1),\, \left(1+2y\epsilon_1\right)F(R_1,h_1,1)\right].$$
    By  Lemma~\ref{first-phase-lemma},
    with probability at least $1-\gamma_1$,
    \begin{eqnarray*}
        F(R_1, h_1,
        1)\left(\sum\limits_{i=1}^{\setCount}m_i\right)\in
        [&&\frac{(1-\epsilon_0){(1-\alpha_L)(1-\beta_L)}}{1+\delta}|A_1\cup\cdots \cup
        A_{\setCount}|,\\
        &&  {(1+\epsilon_0)(1+\alpha_R)(1+\beta_R)(1+\delta)}|A_1\cup\cdots \cup
        A_{\setCount}|].
    \end{eqnarray*}



    Therefore, with probability at least
    $1-\gamma_1-2y\gamma_2$,
    \begin{eqnarray*}
        sum\cdot M\in
        [&&\frac{(1-y\epsilon_1)(1-\epsilon_0)(1-\epsilon_1)(1-\alpha_L)(1-\beta_L)}{1+\delta}\cdot
        |A_1\cup\cdots \cup A_{\setCount}|,\\
        &&{(1+2y\epsilon_1)(1+\epsilon_0)(1+\epsilon_1)(1+\alpha_R)(1+\beta_R)(1+\delta)}|A_1\cup\cdots
        \cup A_{\setCount}|].
    \end{eqnarray*}
    Now assume
    \begin{eqnarray*}
        sum\cdot M\in
        [&&\frac{(1-y\epsilon_1)(1-\epsilon_0)(1-\epsilon_1)(1-\alpha_L)(1-\beta_L)}{1+\delta}\cdot
        |A_1\cup\cdots \cup A_{\setCount}|,\\
        &&{(1+2y\epsilon_1)(1+\epsilon_0)(1+\epsilon_1)(1+\alpha_R)(1+\beta_R)(1+\delta)}|A_1\cup\cdots
        \cup A_{\setCount}|].
    \end{eqnarray*}

    By Statement~\ref{sub-add-case} of Lemma~\ref{ex-lemma}, we have
    $$1-\epsilon\le 1-y\epsilon_1-\epsilon_0-\frac{9}{8}\epsilon_1\le
    (1-y\epsilon_1)(1-\epsilon_0)\left(1-\frac{9}{8}\epsilon_1\right)\leq \frac{(1-y\epsilon_1)(1-\epsilon_0)\left(1-\epsilon_1\right)}{1+\delta},$$
    and
    $(1+2y\epsilon_1)(1+\epsilon_0)(1+\epsilon_1)(1+\delta)\le
    (1+2y\epsilon_1)(1+2\epsilon_0+\epsilon_1)(1+\delta)\le
    (1+2y\epsilon_1)(1+4\epsilon_0+2\epsilon_1+\delta)\le
    (1+8\epsilon_0+4\epsilon_1+2\delta+2y\epsilon_1)\le
    (1+8\epsilon_0+4\epsilon_1+\epsilon_2+2y\epsilon_1)\le
    (1+8\epsilon_0+{\epsilon_0\over 3}+{\epsilon_0\over
        3}+{\epsilon_0\over 3})\le 1+9\epsilon_0\le 1+\epsilon$.
    Therefore,
    $$sum\cdot M\in
    [{(1-\epsilon)(1-\alpha_L)(1-\beta_L)}\cdot |A_1\cup\cdots \cup
    A_{\setCount}|, {(1+\epsilon)(1+\alpha_R)(1+\beta_R)}\cdot
    |A_1\cup\cdots \cup A_{\setCount}|].$$

    The algorithm may fail at the case after selecting $R_1$, or one of
    the stages. By the union bound, the failure probability is at most
    $\gamma_1+2\gamma_2\cdot \log {\setCount}\le \gamma$. We have that
    with probability at least $1-\gamma$ to output the sum that
    satisfies the accuracy described in the theorem.
    The running time and the number of rounds of the algorithm follow
    from Lemma~\ref{time-lemma} and Lemma~\ref{stage-lemma},
    respectively.
\end{proof}

Since $1\le z_{\min}\le \Thick(L)\le \ThickMax(L)\le z_{\max}\le
\setCount$, we have the following Corollary~\ref{nearlinear-coro}.
Its running time is almost linear in the classical model.

\begin{corollary}\label{nearlinear-coro}
    There is a $\bigO(\poly({1\over \epsilon},\log{1\over \gamma})\cdot
    {\setCount}\cdot (\log {\setCount})^{\bigO(1)})$ time and
    $\bigO(\log \setCount)$ rounds algorithm for $|A_1\cup A_2\cup\cdots
    A_{\setCount}|$  such that with probability at least $1-\gamma$, it
    gives a $sum\cdot M\in [{(1-\epsilon)(1-\alpha_L)(1-\beta_L)}\cdot
    |A_1\cup\cdots \cup A_{\setCount}|,
    {(1+\epsilon)(1+\alpha_R)(1+\beta_R)}\cdot |A_1\cup\cdots \cup
    A_{\setCount}|]$.
\end{corollary}

\begin{proof}
    We let $\fFive(m)=8$  with $c_1=0$ in equation~(\ref{f5-def-eqn}).
    Let $z_{min}=1$ and $z_{max}=\setCount$. It follows from
    Theorem~\ref{nonadaptive-thm} and Statements~\ref{f-1-line} and~\ref{f-6-line} of
    Lemma~\ref{basic-facts-lemma} as
    we have the inequality (\ref{final-ineqn-coro1}):
    \begin{eqnarray}
      \left({{\setCount}\fOne(\setCount)\fTwo(\setCount)\over
        z_{\min}}\cdot \left({\log {z_{\max}\over z_{\min}}\over \log
        \fFive(\setCount)}\right)\right)=\bigO(\poly({1\over \epsilon},\log{1\over \gamma})\cdot
    {\setCount}\cdot (\log {\setCount})^{\bigO(1)}).\label{final-ineqn-coro1}
    \end{eqnarray}
\end{proof}

\begin{corollary}\label{nearlinear-coro2}For each $\xi>0$, there is a $\bigO(\poly({1\over \epsilon},\log{1\over
            \gamma})\cdot \setCount^{1+\xi})$ time and $\bigO({1\over \xi})$
    rounds algorithm for $|A_1\cup A_2\cup\cdots A_{\setCount}|$ such
    that with probability at least $1-\gamma$, it gives a $sum\cdot M\in
    [{(1-\epsilon)(1-\alpha_L)(1-\beta_L)}\cdot |A_1\cup\cdots \cup
    A_{\setCount}|, {(1+\epsilon)(1+\alpha_R)(1+\beta_R)}\cdot
    |A_1\cup\cdots \cup A_{\setCount}|]$.
\end{corollary}

\begin{proof}
    We let $\fFive(m)=8\setCount^{\xi/2}$ with $c_1={\xi\over 2}$ in
    equation~(\ref{f5-def-eqn}). Let $z_{min}=1$ and
    $z_{max}=\setCount$. It follows from Theorem~\ref{nonadaptive-thm}
    and Statements~\ref{f-1-line} and~\ref{f-6-line} of Lemma~\ref{basic-facts-lemma} as
    we have the inequality (\ref{final-ineqn-coro2}):
    \begin{eqnarray}
      \left({{\setCount}\fOne(\setCount)\fTwo(\setCount)\over
        z_{\min}}\cdot \left({\log {z_{\max}\over z_{\min}}\over \log
        \fFive(\setCount)}\right)\right)=\bigO(\poly({1\over \epsilon},\log{1\over
            \gamma})\cdot
            \setCount^{1+\xi}).\label{final-ineqn-coro2}
    \end{eqnarray}
\end{proof}

An interesting open problem is to find an $\bigO(\setCount)$ time
and  $\bigO(\log \setCount)$ rounds approximation scheme for
$|A_1\cup A_2\cup\cdots A_{\setCount}|$ with a similar accuracy
performance as Corollary~\ref{nearlinear-coro}. We were not able to
adapt the method from Karp, Luby, and Madras \cite{KarpLubyMadras89}
to solve this problem.

\section{Approximate Random Sampling for Lattice Points in High Dimensional Ball}\label{random-lattice-section}
In this section, we propose algorithms to approximate the numebr of lattice points in a high dimensional ball, and also develop algorithms to generate a random lattice point inside a high dimensional ball.

Before present the algorithms, some definitions are given below.
\begin{definition}\label{definition1}Let integer $d>0$ be a dimensional number, $\mathbb{R}^d$ be the $d-$dimensional Euclidean Space.
    \begin{enumerate}

        \item For two points $p,\ q \in\mathbb{R}^d,$ define $||p-q||$ to be Euclidean Distance.
        \item A point $p\in \mathbb{R}^d$ is a lattice point if $p=(y_1, ..., y_d)$ with $y_i\in \mathbb{Z}$ for $i=1, 2, ..., d.$
        \item Let $p\in\mathbb{R}^d$, and $r>0.$ Define $B_d(r, p, d)$ be a $d-$dimensional ball of radius $r$ with center at $p.$
        \item Let $q=(\mu_1, \mu_2, ..., \mu_d)\in \mathbb{R}^d.$ 
        Define $B_d(r, q, k)=\{(z_1, z_2, ..., z_d)\in\mathbb{R}^d : z_1=\mu_1, ..., z_{d-k}=\mu_{d-k}\  and\ \sum\limits_{i=1}^{d}(\mu_i-z_i)^2\leq r^2\}.$
        \item Let $p\in\mathbb{R}^d$, and $r>0.$ Define $C(r, p, d)$ be
        the number of lattice points in the $d-$dimensional ball of radius $r$ with the
        center at $p$.
        \item Let $\lambda,$ $l$ be real numbers. Define $D(\lambda,d, l)=\{ (x_1,\cdots, x_d) : (x_1,\cdots, x_d)$ with $x_k=i_k+j_k\lambda$ for an integer $j_k\in [-l,\ l]$, and another arbitrary integer $i_k$ for $k=1, 2, ..., d\}$.
        \item Let $\lambda,$ $l$ be real numbers. Define $D^{\ast}(\lambda,d, l)=\{ (x_1,\cdots, x_d) : (x_1,\cdots, x_d)$ with $x_k=j_k\lambda$ for an integer $j_k\in [-l,\ l]$ with $k=1, 2, ..., d\}$.
        \item Let $\lambda=a^{-m},$ where $a$ and $m$ are integer and $a\geq 2.$ Define $D^{\ast\ast}(\lambda,d)=\{(x_1,\cdots, x_d): (x_1,\cdots, x_d)$ with $x_k=i_k+j_k\lambda$ for an integer $j_k\in [-\lambda^{-1}+1, \lambda^{-1}-1],$ and another arbitrary integer $i_k$ for $k=1, 2, ..., d\}$.

    \end{enumerate}
\end{definition}

\subsection{Randomized Algorithm for Approximating Lattice Points for High Dimensional Ball}
In this section, we develop algorithms
to approximate the number of lattice points in a $d$-dimensional ball $B_d(r, p, d)$.
Two subsubsections are discussed below.

\subsubsection{Counting Lattice Points of High Dimensional Ball with Small Radius}

In this section, we develop a dynamic programming algorithm to count the number of lattice points in $d-$dimensional ball $B_d(r, p, d).$ Some definitions and lemmas that is used to prove the performance of algorithm are given before present the algorithm.

\begin{definition}\label{set-lattice-points}
    Let $p$ be a point in $\mathbb{R}^d,$ and $p\in D(\lambda,d, L).$ Define $E(r', p, h, k)$ be the set of $k-$dimensional balls $B_d(r', q, k)$ of radii $r'$ with center at
    $q=(y_1, y_2, ..., y_h, x_{h+1}, ..., x_d)$ where $h=d-k$ is the number of initial integers of the center $q$ and $y_t\in \mathbb{Z}$ for $t=1, 2, ..., h.$
\end{definition}
Lemma~\ref{equivalence-classes} shows that for any two balls with same dimensional number, if their radii equal and the number of initial integers of their center also equal, then they have same number of lattice points.
\begin{lemma}\label{equivalence-classes}
    For two $k-$dimensional balls $B_d(r, q, k)$ and $B_d(r, q', k),$ if  $B_d(r, q, k)\in  E(r, p, h, k)$ and $B_d(r, q', k)\in E(r, p, h, k),$ then $C(r, q, k)=C(r, q', k).$
\end{lemma}

\begin{proof}
    In order to prove that $C(r, q, k)=C(r, q', k),$ we need to show that
    there is a bijection bewtten the set of  of lattice points inside ball $B_d(r, q, k)$ and 
    the set of lattice points inside ball  $B_d(r, q', k),$ where $q=(y_1, y_2, ..., y_h, x_{h+1}, ..., x_d)$
    and
    $q'=(y'_1, y'_2, ..., y'_h, x_{h+1}, ..., x_d)$ with
    $y'_t, y_t\in \mathbb{Z}$ for $t=1, 2, ..., h.$

    Statement $1$: $\forall\ q_1=(z_1, z_2, ..., z_d)\in B_d(r, q, k),$ where $z_t\in \mathbb{Z}$ for $t=1, 2, ..., d.$

    we have
    $$(z_1-y_1)^2+\cdots+(z_h-y_h)^2+(z_{h+1}-x_{h+1})^2+\cdots+(z_d-x_d)^2\leq r^2$$
    then
    $$(z_1+y'_1-y_1-y'_1)^2+\cdots+(z_h+y'_h-y_h-y'_h)^2+(z_{h+1}-x_{h+1})^2+\cdots+(z_d-x_d)^2\leq r^2.$$

    Therefore, there exists a lattice point $(z_1+y'_1-y_1, ..., z_h+y'_h-y_h, z_{h+1}, ..., z_d)\in B_d(r, q', k)$ correspoding to $q_1.$

    Statement $2$: $\forall\ q'_1=(z'_1, z'_2, ..., z'_d) \in B_d(r, q', k),$ where $z'_t\in \mathbb{Z}$ for $t=1, 2, ..., d.$

    we have
    $$(z'_1-y'_1)^2+\cdots+(z'_h-y'_h)^2+(z'_{h+1}-x_{h+1})^2+\cdots+(z'_d-x_d)^2\leq r^2$$
    and
    $$(z'_1-y'_1+y_1-y_1)^2+\cdots+(z'_h-y'_h+y_h-y_h)^2+(z'_{h+1}-x_{h+1})^2+\cdots+(z'_d-x_d)^2\leq r^2.$$

    Therefore, there exists a lattice point $(z'_1-y'_1+y_1, ..., z'_h-y'_h+y_h, z'_{h+1}, ..., z'_d)\in B_d(r, q, k)$ correspoding to $q'_1.$

    Based on above two statements, there exists a bijection between the set of lattice points inside ball $B_d(r, q, k)$ and the set of lattice points inside ball $B_d(r, q', k).$

    Therefore, $C(r, q, k)=C(r, q', k).$
\end{proof}

Lemma~\ref{without-i-lattice-points} shows that we can move ball $B_d(r, q, k)$ by an integer units in every dimension without changing the number of lattice points in the ball.
\begin{lemma}\label{without-i-lattice-points}
    Let $\lambda$ be a real number. For two $k-$dimensional balls $B_d(r, q_1, k)$ and $B_d(r, q_2, k),$ where $q_1=(y_1, y_2, ..., y_{d-k}, x_{d-k+1}, ..., x_d)$, $q_2=(y'_1, y'_2, ..., y'_{d-k}, x'_{d-k+1}, ..., x'_d)$ with $y_t, y'_t\in \mathbb{Z},$ $t=1, 2, ..., d-k,$ and $x_{t'}=i_{t'}+j_{t'}\lambda$, $i_{t'}$ is an integer and $j_{t'}\in[-l, l]$ for $t'=d-k+1, ..., d,$ if $x'_{t'}=j_{t'}\lambda,$ then we have $C(r, q_1, k)=C(r, q_2, k).$
\end{lemma}

\begin{proof}
     Since $B_d(r, q_1, k)\in E(r, p, h, k)$ and $B_d(r, q_2, k)\in E(r, p, h, k)$ with $h=d-k,$ we have $C(r, q_1, k)=C(r, q_2, k)$ via Lemma~\ref{equivalence-classes}.
\end{proof}
We define $R(r, p, d)$ be a set of radii $r'$ for the balls that generated by the intersection of $B_d(r, p, d)$ wiht hyper-plane $x_1=y_1,$ ..., $x_k=y_k,$ ..., $x_d=y_d.$
\begin{definition}\label{definition-R}
    For a $d-$dimensional ball $B_d(r, p, d)$ of radius $r$ with center at $p=(x_1, x_2, ..., x_d)$.

    \begin{enumerate}
        \item  Define $R(r, p, d)=\{r': r'^2=r^2-\sum\limits_{i=1}^{k}(y_i-x_i)^2\ with\  y_i\in \mathbb{Z}\ and\   \sum\limits_{i=1}^{k}(y_i-x_i)^2\leq r^2\ for\ some\ integer\ k\in[1, d]\}$.
    \end{enumerate}
\end{definition}
Lemma~\ref{range-of-radius} shows that we can reduce the cardinality of $R(r, p, d)$ from exponentional to polynomial when setting the element of the ball's center has same type (i.e. $p\in D(\lambda, d, l)$.)
\begin{lemma}\label{range-of-radius}
    Let $B_d(r, p, d)$ be a $d-$dimensional ball of radius $r$ with center at $p,$ where $p\in D^{\ast}(\lambda, d, l),$ then $|R(r, p, d)|\leq 4(r+l|\lambda|)^3l^3d^3$ and
    $R(r, p, d)$ can be generated in $\bigO\left( (r+l|\lambda|)^3l^3d^3\right)$ time.
\end{lemma}
\begin{proof}
    Since $r'^2=r^2-\sum\limits_{i=1}^{k}(y_i-x_i)^2$ for $0\leq k\leq d,$ we have $r'$ as:
    \begin{eqnarray*}
        r'^{2}&&=r^{2}-(y_{1}-j_1\lambda)^{2}-\cdots-(y_{d}-j_d\lambda)^{2}\\
        &&= r^2-[y_{1}^{2}-2y_{1}j_{1}\lambda+j_{1}^{2}\lambda^2]-\cdots-[y_{d}^{2}-2y_{d}j_{d}\lambda+j_d^{2}\lambda^{2}]\\
        &&= r^2-\{y_1^{2}+y_2^{2}+\cdots+y_{d}^{2}\}\\
        &&+\{2y_1j_1+2y_2j_2+\cdots+2y_dj_d\}\lambda\\
        &&-\{j_1^2+j_2^2+j_3^2+\cdots+j_d^2\}\lambda^2.
    \end{eqnarray*}

    Let $R'=\{r'|r'^2=r^2-(x+y\lambda+z\lambda^2)\  with\ x,\ y,\ and\  z\ is\  nonnegative\  integer\}$,  it is easy to see that $r'\in R'$ then  $R\subseteq R'.$

    Let
    \begin{equation*}\label{productt}
    \left\{
    \begin{array}{lr}
    X=\{x'|x'=y_1^{2}+y_2^{2}+...+y_{d}^{2}\ with\ y_i\in[r-l|\lambda|, r+l|\lambda|],\ 0\leq i\leq d\} &  \\
    Y=\{y'|y'=2y_1j_1+2y_2j_2+...+2y_dj_d\ with\ y_ij_i\in[I(r-l|\lambda|), I(r+l|\lambda|)],\ 0\leq i\leq d\}& \\
    Z=\{z'|z'=j_1^2+j_2^2+j_3^2+...+j_d^2\ with\ j_i\in[-l, l],\ 0\leq i\leq d\}, &
    \end{array}
    \right.
    \end{equation*}
    then we have:
    \begin{equation}\label{product}
    \left\{
    \begin{array}{lr}
    |Z|\leq dl^2 &  \\
    |Y|\leq 4d(r+l|\lambda|)l& \\
    |X|\leq d(r+l|\lambda|)^2. &
    \end{array}
    \right.
    \end{equation}

    For each $r'\in R$, we have $r'^2=r^2-(x+y\lambda+z\lambda^2)$ with $x\in X$, $y\in Y$, and $z\in Z$.
    Therefore, $|R|\leq dl^2\cdot4d(r+l|\lambda|)l\cdot d(r+l|\lambda|)^2=4(r+l|\lambda|)^3l^3d^3$ via inequality~(\ref{product}). Then $R(r, p, d)$ can be generated in $\bigO\left( (r+l|\lambda|)^3l^3d^3\right)$ time.

\end{proof}

Lemma~\ref{range-of-radius-pi} is a spacial case of Lemma~\ref{range-of-radius}. It shows that there at most $(r^2+1)a^{2m}$ cases of the radii when the elements of the center are the type like fractions in base $a$. For example, $p=(3.891, 5.436, ..., 5.743)\in \mathbb{R}^d.$

\begin{lemma}\label{range-of-radius-pi}
    Let $\lambda=a^{-m}$ where $a$ is a interger with  $a\geq 2.$ Let $B_d(r, p, d)$ be a $d-$dimensional ball of radius $r$ with center at $p\in D^{\ast\ast}(\lambda,d),$  then $|R(r, p, d)|\leq (r^2+1)a^{2m}$ and
    $R(r, p, d)$ can be generated in $\bigO\left( (r^2+1)a^{2m}\right)$ time.
\end{lemma}
\begin{proof}
    We have
    \begin{eqnarray*}
        r'^{2}&&=r^{2}-(y_{1}-j_1\lambda)^{2}-\cdots-(y_{d}-j_d\lambda)^{2}\\
        &&= r^2-[y_{1}^{2}-2y_{1}j_{1}\lambda+j_{1}^{2}\lambda^2]-\cdots-[y_{d}^{2}-2y_{d}j_{d}\lambda+j_d^{2}\lambda^{2}]\\
        &&= r^2-\{y_1^{2}+y_2^{2}+\cdots+y_{d}^{2}\}\\
        &&+\{2y_1j_1+2y_2j_2+\cdots+2y_dj_d\}\lambda\\
        &&-\{j_1^2+j_2^2+j_3^2+\cdots+j_d^2\}\lambda^2
    \end{eqnarray*}
    via Lemma~\ref{range-of-radius}.

    For each $r'^2$, it can be transformed into $r'^{2}=r^2-(x+y\lambda+z\lambda^2)$ with $x,\, y$ and $z$ are integers,

    and
    \begin{equation}\label{product1}
    \left\{
    \begin{array}{lr}
    |z|\leq a^m &  \\
    |y|\leq a^m& \\
    |x|\leq (r^2+1). &
    \end{array}
    \right.
    \end{equation}

    Therefore, $|R|\leq (r^2+1)a^{2m}$ via inequality~(\ref{product1}). Then $R(r, p, d)$ can be generated in $\bigO\left( (r^2+1)a^{2m}\right)$ time.

\end{proof}

\begin{definition}
    For a $d-$dimensional ball $B_d(r, p, d)$ of radius $r$ with center at $p=(x_1, x_2, ..., x_d)$.

    \begin{enumerate}
        \item Define $p[k]=(0, ..., 0, x_{k+1}, ..., x_{d})$ for some integer $k\in[1, d]$.
        \item Define $Z(r,x,t)$ with $Z(r, x, t)^2=r^2-(t-x)^2$ if $|t-x|\le r,$ where $t$ is a integer and $x\in \mathbb{R}$.
    \end{enumerate}
\end{definition}
We give a dynamic programming algorithm to count the number of lattice points in a $d-$dimensional ball $B_d(r, p , d).$
\begin{algorithm}[H]
    \caption{CountLatticePoint($r$, $p$, $d$)}\label{euclid}
    $\mathbf{Input:}$ $p=(x_1, x_2, ..., x_d)$ where $x_k=i_k+j_k\lambda$ for  an integer $j_k\in [-l, l]$, and another arbitrary integer $i_k$ for $k=1, 2, ..., d.$     $r$ is radius and $d$ is dimensional numbers.\\
    $\mathbf{Output:}$ The number of lattice points of the $d-$dimensional ball $B_d(r, p, d).$
    \begin{algorithmic}[1]
        \State\qquad Let $r_0=r$
        \State\qquad For $k=d-1$ to $0$ \label{line2}
        \State\qquad\qquad for each $r_k\in R(r, p, d)$ \label{line3}
        \State\qquad\qquad\qquad let $C(r_k,p[k],d-k)=\sum\limits_{t \in \mathbb{Z}\ and\ t\in[-r_k+x_{k+1},\  r_k+x_{k+1}]}C(z(r_k, x_{k+1}, t),p[k+1],d-(k+1))$\label{line4}
        \State\qquad\qquad\qquad save $C(r_k,p[k],d-k)$ to the look up table
        \State\qquad Return $C(r_0, p[0], d)$

    \end{algorithmic}
\end{algorithm}

We note that if $d-(k+1)=0$ then $C(z(r_k, x_{k+1}, t),p[k+1],d-(k+1))=1,$ otherwise $z(r_k, x_{k+1}, t)$ is in $R(r, p, d)$ (i.e. $C(z(r_k, x_{k+1}, t),p[k+1],d-(k+1))$ is avaiable in the table).


\begin{theorem}\label{smaller-total-point-lemma}
    Assume $\lambda$ be a real number and $p\in D(\lambda, d, l),$ then
    there is a $\bigO(r(r+l|\lambda|)^3l^3d^4)$ time algorithm to count $C(r,p, d)$.
\end{theorem}

\begin{proof}
    Line~\ref{line2} has $d$ iterations, Line~\ref{line3} takes $4(r+l|\lambda|)^3l^3d^3$ to compute $r_k$ via Lemma~\ref{range-of-radius}, and Line~\ref{line4} has at most $2\lfloor r\rfloor+1$ items to add up.

    Therefore, the algorithm CountLatticePoints(.) takes $\bigO(r(r+l|\lambda|)^3l^3d^4)$ running time.


\end{proof}

\textbf{Remark:} When $\lambda=\frac{1}{\pi},$ this is a specail case of Theorem~\ref{smaller-total-point-lemma}, and the running time of the algorithm is $\bigO(r(r+l|\lambda|)^3l^3d^4).$ The algorithm can count the lattice points of high dimensional ball if the element of the center of the ball has same type like $i+j\lambda$ even though $\lambda$ is a irrational number.

Theorem~\ref{smaller-total-point-lemma-ration} shows that the algorithm can count the number of lattice points of high dimensional ball if the element of the center of the ball has same type like fractions in base $a$.
\begin{theorem}\label{smaller-total-point-lemma-ration}
    Assume $\lambda=a^{-m}$ and $p\in D^{\ast\ast}(\lambda,d)$, where $m$ and $a$ are integers with $a\geq 2,$ then
    there is a $\bigO(r^3a^{2m}d)$ time algorithm to count $C(r,p, d)$.
\end{theorem}
\begin{proof}
    Line~\ref{line2} has $d$ iterations, Line~\ref{line3} takes $(r^2+1)a^{2m}$ to compute $r_k$ via Lemma~\ref{range-of-radius-pi}, and Line~\ref{line4} has at most $2\lfloor r\rfloor+1$ items to add up.

    Therefore, the algorithm CountLatticePoints(.) takes $\bigO(rd(r^2+1)a^{2m})$ running time.

\end{proof}

\begin{corollary}\label{}
    Assume $\lambda=10^{-m}$ and $p\in D^{\ast\ast}(\lambda, d)$, where $m$ is a integer, then
    there is a $\bigO(r^310^{2m}d)$ time algorithm to count $C(r,p, d)$.
\end{corollary}

\subsubsection{Approximating Lattice Points in High Dimensional Ball with Large Radius}
In this section, we present an $(1+\beta)$-approximation algorithm to approximate the number of lattice points in a $d-$dimensional ball $B_d(r, p, d)$ of large radius with an arbitrary center $p$, where $\beta$ is used to control the accuracy of approximation.

Some definitions are presented before prove theorems.
\begin{definition}\label{define-e-i}
    For each lattice point $q=(y_1, y_2,..., y_d)\in \mathbb{R}^d$ with $y_i\in \mathbb{Z}$ for $i=1, 2, ..., d$.
    \begin{enumerate}
        \item Define $Cube(q)$ to be the $d-$dimensional unit cube with center at $\left(y_1+\frac{1}{2}, ..., y_d+\frac{1}{2}\right).$
        \item Define $I(B_d(r, p, d))=\{q\ |\ Cube(q)\subset B_d(r, p, d) \}.$
        \item Define $E(B_d(r, p, d))=\{q\ |\ Cube(q) \notin I(B_d(r, p, d))\ and\ Cube(q) \cap B_d(r, p, d)\neq \emptyset\}.$
    \end{enumerate}
\end{definition}

Theorem~\ref{large-total-point-lemma-ration} gives an $(1+\beta)-$approximation with running time $\bigO(d)$ algorithm to approximate the number of lattice point $C(r, p, d)$ with $p$ is an arbitrary center and $r>\frac{2d^{\frac{3}{2}}}{\beta}.$
\begin{theorem}\label{large-total-point-lemma-ration}
    For an arbitrary $\beta\in (0, 1),$ there is a $(1+\beta)-$approximation algorithm to compute $C(r, p, d)$ of $d-$dimensional ball $B_d(r, p, d)$ with running time $\bigO(d)$ for an arbitrary center $p$ when $r>\frac{2d^{\frac{3}{2}}}{\beta}.$
\end{theorem}
\begin{proof}
    Let $|I(B_d(r, p, d))|$ be the number of lattice points $q\in I(B_d(r, p, d))$, $|E(B_d(r, p, d))|$ be the number of lattice points $q\in E(B_d(r, p, d)),$ and $V_d(r)$ be the volume of a $d-$dimensional ball with radius $r$.

    Now consider two $d-$dimensional balls $B_d(r-\sqrt{d}, p, d)$ and $B_d(r+\sqrt{d}, p, d)$ that have the same center as ball $B_d(r, p, d).$ Since every lattice point $q$ corresponds to a $Cube(q)$ via Definition~\ref{define-e-i}, 
    then we have:
    \begin{equation*} \label{}
    \left\{
    \begin{array}{lr}
    V_{d}(r-\sqrt{d})\leq |I(B_d(r, p, d))| \leq V_{d}(r)&  \\
    0\leq |E(B_d(r, p, d))|\leq V_{d}(r+\sqrt{d})-V_{d}(r).&
    \end{array}
    \right.
    \end{equation*}
    Therefore,
    $$V_{d}(r-\sqrt{d})\leq C(r, p, d)=|I(B_d(r, p, d))|+|E(B_d(r, p, d))|\leq V_{d}(r+\sqrt{d}).$$
    Then the bias is $\frac{|I(B_d(r, p, d))|+|E(B_d(r, p, d))|}{V_{d}(r)}$ when using $V_{d}(r)$ to approximate $C(r, p, d).$

    The volume formula for a $d-dimensional$ ball of raduis $r$ is
    $$V_d(r)=f(d)\cdot r^d$$
    where $f(d)=\pi^{\frac{d}{2}}\Gamma\left(\frac{1}{2}d+1\right)^{-1}$ and $\Gamma(.)$ is Euler's gamma function. Then
    \begin{eqnarray*}
    \frac{|I(B_d(r, p, d))|+|E(B_d(r, p, d))|}{V_{d}(r)}&&\leq\frac{V_{d}(r+\sqrt{d})}{V_{d}(r)} \\
    &&=\frac{f(d)\cdot(r+\sqrt{d})^d}{f(d)\cdot r^d} \\
    &&=\left(1+\frac{\sqrt{d}}{r}\right)^d\\
    &&\leq e^{\frac{d^{\frac{3}{2}}}{r}}\\
    &&\leq 1+\frac{2d^{\frac{3}{2}}}{r} \label{radius-decide11}.
    \end{eqnarray*}
    Similarly, we have
    \begin{eqnarray*}
    \frac{|I(B_d(r, p, d))|+|E(B_d(r, p, d))|}{V_{d}(r)}&&\geq\frac{V_{d}(r-\sqrt{d})}{V_{d}(r)} \\
    &&=\frac{f(d)\cdot(r-\sqrt{d})^d}{f(d)\cdot r^d} \\
    &&=\left(1-\frac{\sqrt{d}}{r}\right)^d\\
    &&\geq 1-\frac{d^{\frac{3}{2}}}{r}\\\label{radius-decide22}
    &&\geq 1-\frac{2d^{\frac{3}{2}}}{r}.
    \end{eqnarray*}
    From above two inequalities, we have
    $$\left(1-\frac{2d^{\frac{3}{2}}}{r}\right)\cdot V_{d}(r)\leq C(r, p, d)\leq \left(1+\frac{2d^{\frac{3}{2}}}{r}\right)\cdot V_{d}(r),$$
    then we have
    $$\frac{1}{1+\frac{2d^{\frac{3}{2}}}{r}}\cdot C(r, p, d)\leq V_{d}(r)\leq \frac{1}{1-\frac{2d^{\frac{3}{2}}}{r}}\cdot C(r, p, d).$$

    Simplify the above inequality, we have
    $$
    \left(1-\frac{2d^{\frac{3}{2}}}{r-2d^{\frac{3}{2}}}\right)C(r, p, d)\leq V_{d}(r)\leq\left(1+\frac{2d^{\frac{3}{2}}}{r-2d^{\frac{3}{2}}}\right)C(r, p, d).
    $$
    Thus, we have
    \begin{equation}\label{runnning-time}
    (1-\beta)C(r, p, d)\leq V_{d}(r)\leq(1+\beta)C(r, p, d)
    \end{equation}
    with $\beta>\frac{2d^{\frac{3}{2}}}{r-2d^{\frac{3}{2}}}.$


    It takes $\bigO(d)$ to compute $V_d(r)=f(d)\cdot r^d$, since it takes $\bigO(d)$ to compute $f(d)$ where $f(d)=\pi^{\frac{d}{2}}\Gamma\left(\frac{1}{2}d+1\right)^{-1}.$ Therefore, the algorithm takes $\bigO(d)$ running time to approximate $C(r, p, d)$
    becasue of Equation~(\ref{runnning-time}).
\end{proof}

\begin{theorem}\label{theorem14}
    There is an $(1+\beta)$-approximation algorithm with running time $\bigO(d)$ to approximate $C(r, p, d)$ of $B_d(r, p ,d)$ with an arbitrry center $p$ when $r>\frac{2d^{\frac{3}{2}}}{\beta}$; and there is an dynamic programming algorithm with running time 
    $\complexityLatticePoint$ to count $C(r, p, d)$ with center $p\in D(\lambda,
    d, l)$  when $r\leq \frac{2d^{\frac{3}{2}}}{\beta}.$
\end{theorem}

\begin{proof}
    We discuss two cases based the radius of the $d$-dimensional ball.

    Case 1: When counting the number of
    lattice points of a $d$-dimensional ball with center $p \in D(\lambda,
    d, l)$ for $r\leq \frac{2d^{\frac{3}{2}}}{\beta}$, apply Theorem~ \ref{smaller-total-point-lemma}.

    Case 2: When approximating the number of
    lattice points of a $d$-dimensional ball with an arbitrary center $p$ for $r>\frac{2d^{\frac{3}{2}}}{\beta}$, apply Theorem~ \ref{large-total-point-lemma-ration}.
\end{proof}

\begin{corollary}
    There is a dynamic programming algorithm to count $C(r, p, d)$ of $B_d(r, p ,d)$ with running time $\complexityLatticePoint$  for $p\in D(\lambda,
    d, l)$ when $r\leq \frac{2d^{\frac{3}{2}}}{\beta}.$
\end{corollary}

\subsection{A Randomized Algorithm for Generating Random Lattice Point of High Dimensional Ball}
In this section, we propose algorithms to generate a random lattice point inside a high dimensional ball. Two subsections are discussed below.

\subsubsection{Generating a Random Lattice Point inside High Dimensional Ball with Small Radius}

In this section, we develop a recursive algorithm to generate a random lattice point inside a $d-$dimensional ball $B_d(r, p, d)$ of small radius with center $p\in D(\lambda, d, l).$

The purpose of the algorithm
RecursiveSmallBallRandomLatticePoint$(r, p, t, d)$ is to recursively
generate a random lattice point in the ball $B_d(r, p, t)$.
\begin{algorithm}[H]
    \caption{RecursiveSmallBallRandomLatticePoint(r, p, t, d)}\label{}
    $\mathbf{Input:}$ $p=(y_1, y_2,...,y_{d-t}, x_{d-t+1}, ..., x_d)$ where
    $x_k=i_k+j_k\lambda$ with arbitrary integer $i_k,$ integer $j_k\in
    [-l, l],$ and $y_i\in \mathbb{Z}, i=1,2,...,d-t$,  $t$ is a
    dimension number with $0\leq t\leq d.$\\
    $\mathbf{Output:}$ Generate a random lattice point inside $t-$dimensional ball.
    \begin{algorithmic}[1]
        \State\qquad Save $C(r_k,p[k],d-k)$ into look up table C-Table by using Algorithm $CountLatticePoint(r, p, d)$ for $k=0, 1, ..., d-1$
        \State\qquad If $t=0$
        \State \qquad\qquad Return lattice point $(y_1, y_2,..., y_d)$
        \State\qquad Else
        \State\qquad\qquad Return RecursiveSmallBallRandomLatticePoint$(r', q, t-1,
        d)$ with probability
        $\frac{C(r', q, t-1)}{C(r, p, t)},$
        where $q=(y_1, y_2,...,y_{d-t},y_{d-t+1}, x_{d-t+2}, ..., x_d)$ with
        $y_{d-t+1}\in[x_{d-t+1}-r, x_{d-t+1}+r]$ satisfying $||p-q||^2\le
        r^2$,
        and $r'^2=r^2-||p-q||^2$\label{algorithm2line4}
    \end{algorithmic}
\end{algorithm}
We note that $C(., ., .)$ is available at C-Table in $\bigO(1)$ step and the implementation of line~\ref{algorithm2line4} of the algorithm is
formally defined below: Partition $I=[1,\ C(r, p, t)]\cap \mathbb{Z}$
into
$I_1,\cdots, I_w$, where $I_i$ is uniquely corresponds to an integer
$y_{d-t+1}\in[x_{d-t+1}-r, x_{d-t+1}+r]$ satisfying  $q=(y_1, y_2,...,y_{d-t},y_{d-t+1}, x_{d-t+2}, ..., x_d)$, $||p-q||^2\le
r^2$, and $|I_i|=C(r', q, t-1)$. Generate a random number $z\in I$.
If $z\in I_i$ ($I_i$ is mapped to $y_{d-t+1}$), then it returns
RecursiveSmallBallRandomLatticePoint$(r', q, t-1, d)$ with $q=(y_1, y_2,...,y_{d-t},y_{d-t+1}, x_{d-t+2}, ..., x_d)$.

The algorithm RandomSmallBallLatticePoint$(r, p, d)$ is to generate
a random lattice point in the ball $B_d(r, p, d)$. It calls the
function RecursiveSmallBallRandomLatticePoint(.).
\begin{algorithm}[H]
    \caption{RandomSmallBallLatticePoint(r, p,
        d)}\label{}
    $\mathbf{Input:}$ $p=(x_1, x_2,..., x_d)$ where $x_k=i_k+j_k\lambda$ with
    arbitrary integer $i_k,$ integer $j_k\in [-l, l]$ for $k=1,2,\cdots,
    d$.\\
    $\mathbf{Output:}$ Generate a random lattice point inside $d-$dimensional ball.
    \begin{algorithmic}[1]
        \State\qquad Return RecursiveSmallBallRandomLatticePoint$(r, p, d, d)$
    \end{algorithmic}
\end{algorithm}
\begin{theorem}\label{smaller-radius-lattice-point}
    For an arbitrary $\beta \in (0, 1),$ assume $\lambda$ be a real number and $p\in D(\lambda, d, l),$ then
    there is a $\complexityLatticePoint$ time algorithm to generate a lattice point inside a $d-$dimensional ball $B_d(r, p, d).$ 
\end{theorem}

\begin{proof}
    By algorithm RandomSmallBallLatticePoint(.), we can generate a random lattice point inside $d-$dimensional ball $B_d(r, p, d)$ with probability $\frac{C(r', q, d-1)}{C(r, p, d)}\cdot\frac{C(r'', q', d-2)}{C(r', q, d-1)}\cdot...\cdot\frac{1}{C(r^{(d-1)}, q^{(d-1)}, 0)}=\frac{1}{C(r, p, d)}.$

    It takes $\complexityLatticePoint$ to compute $C(r, p, d)$ via Theorem~\ref{theorem14},
    then algorithm SmallBallRandomLatticePoint(.) takes $\complexityLatticePoint+\bigO(d)$ running time. Thus, the algorithm takes $\complexityLatticePoint$ running time.
\end{proof}

\subsubsection{Generating a Random Lattice Point of High Dimensional Ball with Large Radius}
In this section, we develop an $(1+\alpha)-$approximation algorithm to generate a random lattice point inside a $d-$dimensional ball $B_d(r, p, d)$ of large radius $r$ with arbitrary center $p,$ where $\alpha$ is used to control the accuracy of approximation.

We first propose an approximation algorithm RecursiveBigBallRandomLatticePoint(.) to generate a random lattice point inside a $d-$dimensional ball $B_d(r, p, d)$ of radius $r$ with lattice point center $p,$ then we apply algorithm RecursiveBigBallRandomLatticePoint(.) to design algorithm BigBallRandomLatticePoint(.) to generate an approximate random lattice point in a $d-$dimensional ball $B_d(r', p, d)$ of radius $r'$ with arbitrary center $p.$

Before present the algorithms, we give some definition and lemmas that is used to analysis algorithm RecursiveBigBallRandomLatticePoint(.).

\begin{definition}\label{definition-pro}
    For an arbitrary $\beta\in (0, 1),$ let $B_d(r, q, k)$ be $k-$dimensional ball of radius $r$ with arbitrary center $q.$  Define $P(r, q, k)$ as

    $$ P(r, q, k)=\left\{
    \begin{array}{rcl}
    C(r, q, k)     &      & {r\leq \frac{2d^{\frac{3}{2}}}{\beta}}\\
    V_k(r)       &      & {otherwise,}
    \end{array} \right. $$
    where $C(r, q, k)$ is the number of lattice point of $k-$dimensional ball $B_d(r, q, k)$ and $  V_k(r)$ is the volume of ball $B_d(r, q, k).$
\end{definition}
Lemma~\ref{appro-lattice-center-lattice-point} shows that we can use $P(r, q, k)$ to approximate $C(r, q, k)$ for $k-$dimensional ball $B_d(r, q, k)$ no matter how much the radius $r$ it is.

\begin{lemma}\label{appro-lattice-center-lattice-point}
    For an arbitrary $\beta\in (0, 1).$ Let $B_d(r, q, k)$ be $k-$dimensional ball of radius $r$ with arbitrary center $q,$ then $(1-\beta)C(r, q, k)\leq P(r, q, k)\leq (1+\beta)C(r, q, k).$ 
\end{lemma}

\begin{proof}
    Two cases are considered.

    Case 1: If $r\leq \frac{2d^{\frac{3}{2}}}{\beta},$ we have $P(r, q, k)= C(r, q, k)$ via Definition~\ref{definition-pro}.

    Case 2: If $r> \frac{2d^{\frac{3}{2}}}{\beta},$
    we have:
    $$(1-\beta)\cdot C(r, q, k)\leq V_k(r)\leq (1+\beta)\cdot C(r, q, k)$$
    via Theorem~\ref{large-total-point-lemma-ration},
    where $V_k(r)$ be the volume of $k-$dimensional ball $B_d(r, q, k)$ with radius $r.$

    Therefore, we have
    $$(1-\beta)\cdot C(r, q, k)\leq P(r, q, k)\leq (1+\beta)\cdot C(r, q, k),$$
    because $P(r, q, k)= V_k(r)$ via Definition~\ref{definition-pro}.

    By combining the above two cases, we conclude that:
    $$(1-\beta)C(r, q, k)\leq P(r, q, k)\leq (1+\beta)C(r, q, k).$$
\end{proof}
Lemma~\ref{same-lattice-center-lattice-point} shows that for two $k-$dimensional balls, if their radius are almost equal, then the number of their lattice points also are almost equal.
\begin{lemma}\label{same-lattice-center-lattice-point}
    For an arbitrary $\beta\in (0, 1)$ and a real number $\delta,$ let $B_d(r', q, k)$ be a $k-$dimensional ball of radius $r'$ with lattice center at $q$ and $B_d(r'', q, k)$ be a $k-$dimensional ball of radius $r''> \frac{2d^{\frac{3}{2}}}{\beta}$ with lattice center at $q$, where $q=(y_1, y_2, ..., y_d)$ with $y_t\in \mathbb{Z}$ and $t=1, 2, ..., d,$ if $r''\leq r'\leq \left(1+\delta\right)r'',$ then  $C(r'', q, k)\leq C(r', q, k)\leq \frac{1+\beta}{1-\beta}\left(1+\delta\right)^kC(r'', q, k).$
\end{lemma}

\begin{proof}
    Let $V_d(r)$ be the volume of $d-$dimensional ball of radius $r.$ Since the volume formula for a $d-dimensional$ ball of raduis $r$ is
    $$V_d(r)=f(d)\cdot r^d$$
    where $f(d)=\pi^{\frac{d}{2}}\Gamma\left(\frac{1}{2}d+1\right)^{-1}$ and $\Gamma(.)$ is Euler's gamma function. Then, we have the following as:
    $$V_k(r'')\leq V_k(r')\leq V_k(r'')\cdot\left(1+\delta\right)^k.$$

    Since $r''> \frac{2d^{\frac{3}{2}}}{\beta},$ $r'\geq r''>2\frac{d^{\frac{3}{2}}}{\beta},$
    then we have
    \begin{equation} \label{relationship-c-v}
    \left\{
    \begin{array}{lr}
    \frac{1}{1+\beta}V_k(r')\leq C(r', q, k)\leq \frac{1}{1-\beta}V_k(r')&  \\
    \frac{1}{1+\beta}V_k(r'')\leq C(r'', q, k)\leq \frac{1}{1-\beta}V_k(r'')&
    \end{array}
    \right.
    \end{equation}
    via Theorem~\ref{large-total-point-lemma-ration},

    Plugging inequality~(\ref{relationship-c-v}) to above inequality, then we have
    \begin{eqnarray*}
    C(r', q, k)&&\leq \frac{1}{1-\beta}V_k(r') \\
    &&\leq\frac{1}{1-\beta}V_k(r'')\cdot\left(1+\delta\right)^k\\
    &&=\frac{(1+\beta)}{(1-\beta)}\frac{1}{(1+\beta)}V_k(r'')\cdot\left(1+\delta\right)^k\\
    &&\leq \frac{(1+\beta)}{(1-\beta)}\cdot\left(1+\delta\right)^kC(r'', q, k) \label{inequlity4}
    \end{eqnarray*}
    and we also have
    $$C(r', q, k)\geq C(r'', q, k).$$

    Therefore,
    $$C(r'', q, k)\leq C(r', q, k)\leq \frac{1+\beta}{1-\beta}\left(1+\delta\right)^kC(r'', q, k).$$
\end{proof}

\begin{definition}
    For an integer interval $[a,b]$, $c\in \mathbb{Z}$, $r>0$, and
    $\delta\in (0,1)$, an $(r,c, 1+\delta)$-partition for $[a,b]$ is to
    divide $[a,b]$ into $[a_1,b_1], [a_2,b_2],\cdots, [a_w, b_w]$ that
    satisfies the following conditions:
    \begin{enumerate}
        \item
        $a_1=a, a_{i+1}=b_i+1$ for $i=1,\cdots, w-1$.
        \item
        For any $x, y\in \{a_i, b_i\}$,  $r^2-(x-c)^2\le
        (1+\delta)^2(r^2-(y-c)^2)$ and $r^2-(y-c)^2\le
        (1+\delta)^2(r^2-(x-c)^2)$.
        \item
        For any $x\in \{a_i, b_i\}$ and $y\in \{a_{i+1},b_{i+1}\}$,
        $r^2-(x-c)^2> (1+\delta)^2(r^2-y^2)$ or $r^2-(y-c)^2>
        (1+\delta)^2(r^2-x^2)$.
    \end{enumerate}
\end{definition}

The purpose of the algorithm  RecursiveBigBallRandomLatticePoint(.) is to recursivly generate a random lattice point inside the $d-$dimensional ball $B_d(r, p, d)$ of radius $r$ with lattice point center $p.$
\begin{algorithm}[H]
    \caption{ RecursiveBigBallRandomLatticePoint(r, p, t, d)}\label{}
    $\mathbf{Input:}$ $p=(z_1, z_2, ..., z_{d-t}, y_{d-t+1}, ..., y_d)$ where
    $z_i\in \mathbb{Z}$ with $1\leq i\leq d-t$, and $y_i\in \mathbb{Z}$
    with $d-t+1\leq i\leq d$, $\alpha\in (0, 1)$ is a parameter to
    control the bias, $r$ is radius, and $t$ is dimensional number.\\
    $\mathbf{Output:}$ $Z=\{z_1, ..., z_d\}.$
    \begin{algorithmic}[1]
        \State\qquad If $t=0$
        \State\qquad\qquad Return $(z_1, z_2, ..., z_d)$
        \State\qquad Let $I_1=[a_1, b_1],\cdots, I_w=[a_w, b_w]$ be the union of intervals via $\left(r, y_{d-t+1}, 1+\frac{\epsilon_4}{g(d)}\right)$-partitions
        for $[\lceil{y_{d-t+1}-r}\rceil, y_{d-t+1}]\cap \mathbb{Z}$ and
        $[y_{d-t+1}+1, \lfloor{y_{d-t+1}+r}\rfloor]\cap \mathbb{Z}$, where $\epsilon_4\in (0, 1)$ and $g(d)$ is a function of $d$ \label{algorithm3line3}
        \State\qquad Let $M=\sum\limits_{i=1}^w (b_i-a_i+1) P(r_{i}, p_{i}, t-1)$,
        where $p_i=(z_1, z_2, ...,z_{d-t}, b_i, y_{d-t+2}, ..., y_d)$, and $r_i^2=r^2-(b_i-y_{d-t+1})^2$
        \State\qquad Return RecursiveBigBallRandomLatticePoint$(r_{i}', p_{i}',
        t-1,d)$ with probability
        $\frac{P(r_{i}, p_{i}, t-1)}{M}$,
        where
        $z_{d-t+1}=b_{i}$, $p_{i}=(z_1, z_2, ...,z_{d-t}, z_{d-t+1}, y_{d-t+2}, ..., y_d)$, and $r_{i}^2=r^2-(z_{d-t+1}-y_{d-t+1})^2$,
        $p_{i}'=(z_1, z_2, ...,z_{d-t}, z_{d-t+1}', y_{d-t+2}, ..., y_d)$, and
        $r_{i}'^2=r^2-(z_{d-t+1}'-y_{d-t+1})^2$ and a random integer $z_{d-t+1}'\in [a_{i}, b_{i}]$\label{algorithm3line5}
    \end{algorithmic}
\end{algorithm}

We note that the implementation of $\left(r, y_{d-t+1}, 1+\frac{\epsilon_4}{g(d)}\right)$-partitions in line~\ref{algorithm3line3} is as the following pictures:
\begin{figure}[H]
    \centering
    \includegraphics[width=0.6\textwidth]{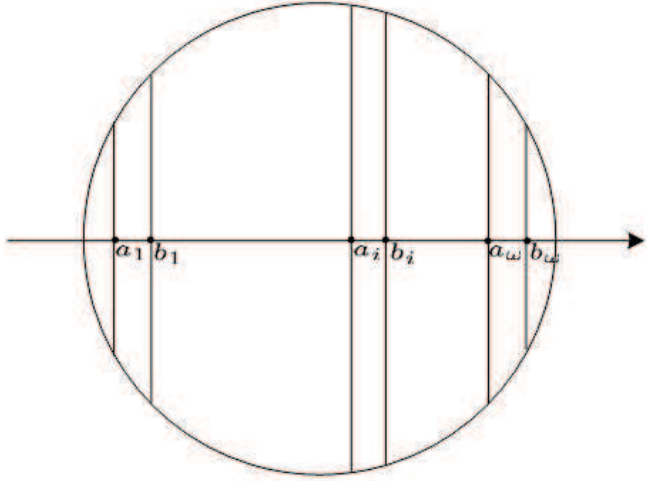} 
    \caption{Example of $\left(r, y_{d-t+1}, 1+\frac{\epsilon_4}{g(d)}\right)$-Partitions in 2D}
\end{figure}

\begin{figure}[H]
    \centering
    \includegraphics[width=0.8\textwidth]{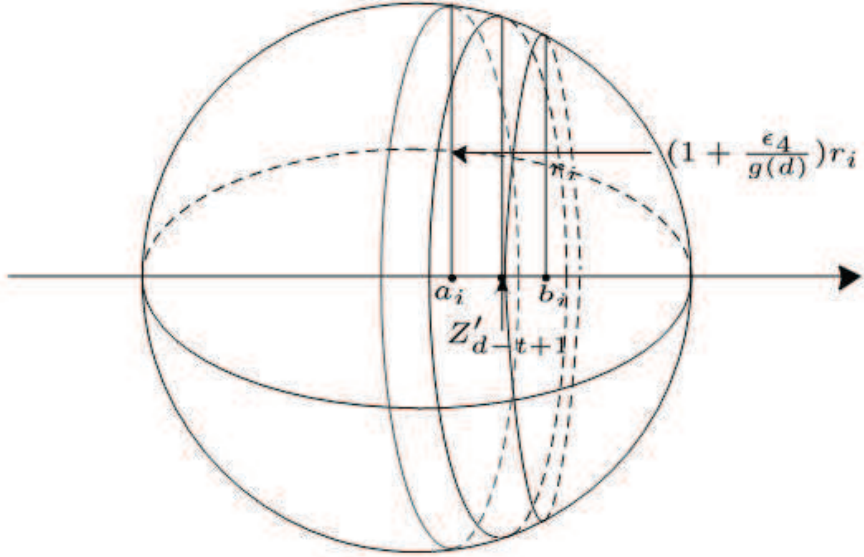} 
    \caption{Example of $\left(r, y_{d-t+1}, 1+\frac{\epsilon_4}{g(d)}\right)$-Partitions in 3D}
\end{figure}

We have the following algorithm that can generate an approximate
random lattice point in a large ball with an arbitrary center, which
may not be a lattice point.
\begin{definition}\label{definiti}Let integer $d>0$ be a dimensional number, $\mathbb{R}^d$ be the $d-$dimensional Euclidean Space.
    \begin{enumerate}
        \item A point $q=(x'_1, x'_2, ..., x'_d)\in \mathbb{R}^d$ is the nearest lattice point of $p=(x_1, ..., x_d)\in \mathbb{R}^d$ if it satisfies $ x'_i=\left\{
        \begin{array}{rcl}
        \lfloor x_i \rfloor     &      & {x_i-[x_i]\leq\frac{1}{2}}\\
        \lceil x_i \rceil      &      & {x_i-[x_i]>\frac{1}{2},}
        \end{array} \right. $ for $x_i\geq 0$ or $ x'_i=\left\{
        \begin{array}{rcl}
        \lceil x_i \rceil     &      & {|x_i|-[|x_i|]\leq\frac{1}{2}}\\
        \lfloor x_i \rfloor      &      & {|x_i|-[|x_i|]<\frac{1}{2},}
        \end{array} \right. $ for $x_i< 0,$ where $i=1, 2, ..., d.$

    \end{enumerate}
\end{definition}

\begin{algorithm}[H]
    \caption{BigBallRandomLatticePoint(r, p, d)}\label{}
    $\mathbf{Input:}$ $p=(x_1,\cdots, ..., x_d)$ where $x_i\in \mathbb{R}$ with
    $1\leq i\leq d,$ $\alpha\in (0, 1)$ is a parameter to control the
    bias, $r$ is radius, and $k$ is dimensional number.\\
    $\mathbf{Output:}$ Generate a random lattice point inside $d-$dimensional ball.
    \begin{algorithmic}[1]
        \State\qquad Let $q$ be the nearest lattice point of $p$ in $\mathbb{R}^d$
        \State\qquad Repeat
        \State\qquad\qquad Let $s=$RecursiveBigBallRandomLatticePoint$(r+\sqrt{d}, q, d)$\label{al3}
        \State\qquad Until $s\in B_d(r, p, d)$
        \State\qquad Return $s$
    \end{algorithmic}
\end{algorithm}

\begin{theorem}\label{lattice-center-lattice-point}
    For an arbitrary $\alpha\in (0, 1),$ there is an algorithm with runing time $\bigO\left(\frac{d^3\log r}{\alpha}\right)$ and $(1+\alpha)-$bias for a $d-$dimensional ball $B_d(r, q, d)$ to generate a random lattice point with radius $r>\frac{2d^3}{\alpha}$ that centered at $q=(y_1, y_2, ..., y_d)$ with $y_t\in \mathbb{Z},$ $t=1, 2, ..., d.$
\end{theorem}

\begin{proof}
     In line~\ref{algorithm3line5} of algorithm RecursiveBigBallRandomLatticePoint(.), define

     $ r'^{2}_i=\left\{
     \begin{array}{rcl}
     r^2-(y_{d-t+1}-a_i)^2     &      & {if\ a_i\leq y_{d-t+1}}\\
     r^2-(y_{d-t+1}-b_i)^2      &      & {otherwise,}
     \end{array} \right. $,

     $ p'_i=\left\{
     \begin{array}{rcl}
     (z_1, z_2, ...,z_{d-t}, a_i, y_{d-t+2}, ..., y_d)     &      & {if\ a_i\leq y_{d-t+1}}\\
     (z_1, z_2, ...,z_{d-t}, b_i, y_{d-t+2}, ..., y_d)      &      & {otherwise,}
     \end{array} \right. $

     and

     $ r^{2}_i=\left\{
     \begin{array}{rcl}
     r^2-(y_{d-t+1}-b_i)^2     &      & {if\ b_i\leq y_{d-t+1}}\\
     r^2-(y_{d-t+1}-a_i)^2      &      & {otherwise,}
     \end{array} \right. $

    $ p_i=\left\{
    \begin{array}{rcl}
    (z_1, z_2, ...,z_{d-t}, b_i, y_{d-t+2}, ..., y_d)     &      & {if\ b_i\leq y_{d-t+1}}\\
    (z_1, z_2, ...,z_{d-t}, a_i, y_{d-t+2}, ..., y_d)      &      & {otherwise.}
    \end{array} \right. $

    Let $v(i)= (b_i-a_i+1),$ and $r'_{i}=\frac{r_{i}}{1+\frac{\epsilon_4}{g(d)}},$ then we have
    $$\sum\limits_{i}C(r'_{i}, p'_{i}, t-1)v(i)\leq C(r_i, p_i, t)\leq \sum\limits_{i}C(r_{i}, p_{i}, t-1)v(i).$$

    Since $r'_{i}=\frac{r_{i}}{1+\frac{\epsilon_4}{g(d)}},$ then
    $$\frac{1-\beta}{1+\beta}\left(1+\frac{\epsilon_4}{g(d)}\right)^{-(t-1)}\sum\limits_{i}C(r_{i}, p_{i}, t-1)v(i)\leq C(r_i, p_i, t)$$
    and
    $$C(r_i, p_i, t)\leq \sum\limits_{i}C(r_{i}, p_{i}, t-1)v(i)$$
    via Lemma~\ref{same-lattice-center-lattice-point}, where $\delta=1+\frac{\epsilon_4}{g(d)}$.

    Via Lemma~\ref{appro-lattice-center-lattice-point} we have
    $$
    \left(\frac{1-\beta}{1+\beta}\right)^2\left(1+\frac{\epsilon_4}{g(d)}\right)^{-(t-1)}\sum\limits_{i}P(r_{i}, p_{i}, t-1)v(i)\leq P(r_i, p_i, t)
    $$
    and
    $$
    P(r_i, p_i, t)\leq \frac{1+\beta}{1-\beta}\sum\limits_{i}P(r_{i}, p_{i}, t-1)v(i).
    $$

    Thus, we have
    $$\left(\frac{1-\beta}{1+\beta}\right)^2\left(1+\frac{\epsilon_4}{g(d)}\right)^{-(t-1)}\leq\frac{P(r_i, p_i, t)}{\sum\limits_{i}P(r_{i}, p_{i}, t-1)v(i)}\leq \frac{1+\beta}{1-\beta}.$$

    From above inequality, we have
    $$\left(\frac{1-\beta}{1+\beta}\right)^2\left(1+\frac{\epsilon_4}{g(d)}\right)^{-(t-1)}\frac{1}{P(r_i, p_i, t)}\leq\frac{1}{\sum\limits_{i}P(r_{i}, p_{i}, t-1)v(i)}\leq \frac{1+\beta}{1-\beta}\frac{1}{P(r_i, p_i, t)}.$$

    Via Lemma~\ref{appro-lattice-center-lattice-point} we have
    $$\frac{(1-\beta)^2}{(1+\beta)^3}\left(1+\frac{\epsilon_4}{g(d)}\right)^{-(t-1)}\frac{1}{C(r_i, p_i, t)}\leq\frac{1}{\sum\limits_{i}P(r_{i}, p_{i}, t-1)v(i)}\leq \frac{1+\beta}{(1-\beta)^2}\frac{1}{C(r_i, p_i, t)}.$$

    Let $g(d)=d^2,$
    $\epsilon_4=\frac{\alpha}{4}$ and $\beta=\frac{\alpha}{\alpha+16d+16}$. Since Algorithm RecursiveBigBallRandomLatticePoint(.) has $d$ iteration, we can generate a random lattice point with bias of probability as:
    \begin{eqnarray*}
    &&\frac{P(r_{i}, p_{i}, d-1)}{\sum\limits_{i}P(r_{i}, p_{i}, d-2)v(i)}\cdot\frac{P(r_{i}, p_{i}, d-2)}{\sum\limits_{i}P(r_{i}, p_{i}, d-1)v(i)}\cdots\frac{P(r_{i}, p_{i}, 0)}{\sum\limits_{i}P(r_{i}, p_{i}, 0)v(i)}\\
    &\leq& \frac{1+\beta}{(1-\beta)^2}\frac{1}{C(r, p, d)}\cdot\left(\frac{1+\beta}{1-\beta}\right)^{d-1}\cdot P(r_i, p_i, 0)\\
    &\leq& \frac{1}{1-\beta}\frac{1}{C(r, p, d)}\cdot\left(\frac{1+\beta}{1-\beta}\right)^{d}\cdot (1+\beta)C(r_i, p_i, 0)\\
    &=& \left(\frac{1+\beta}{1-\beta}\right)^{d+1}\frac{1}{C(r, p, d)}\\
    &=&\left(1+\frac{2\beta}{1-\beta}\right)^{d+1}\frac{1}{C(r, p, d)}\\
    &\leq& e^{\frac{2\beta}{1-\beta}(d+1)}\frac{1}{C(r, p, d)}\\
    &\leq&\left(1+\frac{4\beta}{1-\beta}(d+1)\right)\frac{1}{C(r, p, d)}\\
    &\leq&(1+\alpha)\frac{1}{C(r, p, d)}
    \end{eqnarray*}
    and
    \begin{eqnarray*}
    &&\frac{P(r_{i}, p_{i}, d-1)}{\sum\limits_{i}P(r_{i}, p_{i}, d-2)v(i)}\cdot\frac{P(r_{i}, p_{i}, d-2)}{\sum\limits_{i}P(r_{i}, p_{i}, d-1)v(i)}\cdots\frac{P(r_{i}, p_{i}, 0)}{\sum\limits_{i}P(r_{i}, p_{i}, 0)v(i)}\\
    &\geq&\frac{(1-\beta)^2}{(1+\beta)^3}\left(1+\frac{\epsilon_4}{g(d)}\right)^{-(d-1)}\frac{1}{C(r, p, d)}\left(\frac{1-\beta}{1+\beta}\right)^{2(d-1)}\left(1+\frac{\epsilon_4}{g(d)}\right)^{\frac{-(d-1)(d-2)}{2}}(1-\beta)C(r_i, p_i, 0)\\
    &=&\left(\frac{1-\beta}{1+\beta}\right)^{2d+1}\left(1+\frac{\epsilon_4}{g(d)}\right)^{\frac{-(d-1)d}{2}}\frac{1}{C(r, p, d)}\\
    &\geq&\left(1-\frac{2\beta}{1+\beta}\right)^{2d+1}\left(1+\frac{\epsilon_4}{g(d)}\right)^{\frac{-d^2}{2}}\frac{1}{C(r, p, d)}\\
    &\geq&\left(1-\frac{2\beta}{1+\beta}\right)^{2d}\left(1+\frac{\epsilon_4}{g(d)}\right)^{-d^2}\frac{1}{C(r, p, d)}\\
    &\geq&\left(1-\frac{4\beta d}{1+\beta}\right)\left(1-\frac{\epsilon_4d^2}{g(d)}\right)\frac{1}{C(r, p, d)}\\
    &\geq&\left(1-\frac{4\beta d}{1+\beta}\right)\left(1-\epsilon_4\right)\frac{1}{C(r, p, d)}\\
    &\geq&\left(1-\frac{4\beta d}{1+\beta}-\epsilon_4\right)\frac{1}{C(r, p, d)}\\
    &\geq&\left(1-\frac{4\beta d}{1+\beta}\right)\left(1-\epsilon_4\right)\frac{1}{C(r, p, d)}\\
    &\geq&\left(1-\alpha\right)\frac{1}{C(r, p, d)}.
    \end{eqnarray*}

    Therefore, we can generate a random lattice point with probability between
    $$\left[\left(1-\alpha\right)\frac{1}{C(r, p, d)},\ \left(1+\alpha\right)\frac{1}{C(r, p, d)}\right].$$

    In line~\ref{algorithm3line3} of algorithm RecursiveBigBallRandomLatticePoint(.), it forms a $\left(r, y_{d-t+1}, 1+\frac{\epsilon_4}{g(d)}\right)$-partition
    $I_1,\cdots, I_w$ for
    $[\lceil{y_{d-t+1}-r}\rceil, \lfloor{y_{d-t+1}+r}\rfloor]\cap \mathbb{Z}$ and
    $[y_{d-t+1}+1, \lfloor{y_{d-t+1}+r}\rfloor]\cap \mathbb{Z}$. Then, there are at most $w$ number of $a_i$, where $w$ such that $\frac{r}{\left(1+\frac{\epsilon_4}{g(d)}\right)^{w}}\leq 1.$ Solving $w,$ we have $w\geq \frac{g(d)\log r}{\epsilon_4}.$ And there are $d$ iterations in algorithm RecursiveBigBallRandomLatticePoint(.).

    Thus, the running time of the algorithm is $\bigO\left(\frac{g(d)\log r}{\epsilon_4}\cdot d\right)=\bigO(\frac{d^3\log r}{\epsilon_4})=\bigO\left(\frac{d^3\log r}{\alpha}\right).$
\end{proof}

$\mathbf{Remark:}$ We note that there are at most one $(t-1)-$dimensional ball of radius $r<\frac{2d^3}{\alpha}$ with center at a lattice point, where $t=1, 2, ..., d.$ For this case, we can apply Theorem~\ref{smaller-radius-lattice-point} with $\beta=0.$

\begin{theorem}\label{larger-radius-lattice-point}
    For arbitrary $\alpha\in (0, 1),$ and $\alpha'\in (0, 1),$ there is an $(1+\alpha')-$bias algorithm with runing time $\bigO\left(\frac{d^3\log (r+\sqrt{d})}{\alpha}\right)$ for a $d-$dimensional ball $B_d(r, q, d)$ to generate a random lattice point of radius $r>\frac{2d^{\frac{3}{2}}}{\alpha}$ with an arbitrary center.
\end{theorem}

\begin{proof}
    Consider another ball $B_d(r', q, d)$ of radius $r'$ with lattice center $q=(y_1, y_2, ..., y_d)$ that contains ball $B_d(r, p, d)$ , where $r'=r+\sqrt{d}.$ Let $V_d(r)$ be the volume of a $d-$dimensional ball with radius $r,$ then probability that a lattice point in $B_d(r', q, d)$ belongs to $B_d(r, p, d)$ is at least $(1-\alpha){C(r, p, d)\over C(r', p ,d)}$.

    Via Theorem~\ref{large-total-point-lemma-ration}, we have
    \begin{equation*} \label{relationship-c-v-42}
    \left\{
    \begin{array}{lr}
    \frac{1}{1+\beta}V_d(r)\leq C(r, p, d)\leq \frac{1}{1-\beta}V_d(r)&  \\
    \frac{1}{1+\beta}V_d(r+\sqrt{d})\leq C(r', q, d)\leq \frac{1}{1-\beta}V_d(r+\sqrt{d}),&
    \end{array}
    \right.
    \end{equation*}
    then we have
    $$\frac{1-\beta}{1+\beta}\frac{V_d(r)}{V_d(r+\sqrt{d})}\leq \frac{C(r, p, d)}{C(r', p, d)}\leq \frac{1+\beta}{1-\beta}\frac{V_d(r)}{V_d(r+\sqrt{d})}.$$

    The formula for a $d-dimensional$ ball of raduis $r$ is
    $$V_d(r)=f(d)\cdot r^d$$
    where $f(d)=\pi^{\frac{d}{2}}\Gamma\left(\frac{1}{2}d+1\right)^{-1}$ and $\Gamma(.)$ is Euler's gamma function. Let $\beta=\frac{\alpha}{8+\alpha}$ and $\alpha>\frac{2d^{\frac{3}{2}}}{r+\sqrt{d}},$
    \begin{eqnarray*}
    (1-\alpha){C(r, p, d)\over C(r' ,p ,d)}
    &&\geq (1-\alpha)\frac{1-\beta}{1+\beta}\frac{f(d)\cdot r^d}{f(d)\cdot \left(r+\sqrt{d}\right)^d}\\
    &&=(1-\alpha)\frac{1-\beta}{1+\beta}\left(1-\frac{\sqrt{d}}{r+\sqrt{d}}\right)^d\\
    &&\geq (1-\alpha)\frac{1-\beta}{1+\beta}\left(1-\frac{d^{\frac{3}{2}}}{r+\sqrt{d}}\right)\\ \label{radius-decide2}
    &&\geq (1-\alpha)\left(1-\frac{2\beta}{1-\beta}\right)\left(1-\frac{d^{\frac{3}{2}}}{r+\sqrt{d}}\right)\\
    &&\geq \left(1-\alpha-\frac{2\beta}{1-\beta}-\frac{d^{\frac{3}{2}}}{r+\sqrt{d}}\right).\\
    \end{eqnarray*}
    Therefore, the probability a lattice point in $B_d(r', q, d)$ belongs to $B_d(r, p, d)$ fails is at most $\left(\alpha+\frac{2\beta}{1-\beta}+\frac{d^{\frac{3}{2}}}{r+\sqrt{d}}\right),$ where $\left(\alpha+\frac{2\beta}{1-\beta}+\frac{d^{\frac{3}{2}}}{r+\sqrt{d}}\right)<1,$ which means the algorithm BigBallRandomLatticePoint(.) fails with small possibility.

    The probability to generate a random lattic point in ball $B_d(r', q, d)$ is in range of $$\left[(1-\alpha)\frac{1}{C(r', q, d)}, (1+\alpha)\frac{1}{C(r', q, d)}\right]$$ via Theorem~\ref{lattice-center-lattice-point}. Then the bias to generate a random lattic point in ball $B_d(r, p, d)$ is $\frac{\Pr(p_i)}{\sum_{i}\Pr(p_i)},$ where $\Pr(p_i)\in \left[(1-\alpha)\frac{1}{C(r', q, d)}, (1+\alpha)\frac{1}{C(r', q, d)}\right].$

    Then, we have
    \begin{eqnarray*}
        \frac{\Pr(p_i)}{\sum\limits_{i}\Pr(p_i)}
        &&\leq \frac{(1+\alpha)\frac{1}{C(r', q, d)}}{(1-\alpha)\frac{1}{C(r', q, d)}C(r, p, d)}\\
        &&=\frac{1+\alpha}{1-\alpha}\frac{1}{C(r, p, d)}\\
        &&=\left(1+\frac{2\alpha}{1-\alpha}\right)\frac{1}{C(r, p, d)},\\
    \end{eqnarray*}

    and
    \begin{eqnarray*}
        \frac{\Pr(p_i)}{\sum\limits_{i}\Pr(p_i)}
        &&\geq \frac{(1-\alpha)\frac{1}{C(r', q, d)}}{(1+\alpha)\frac{1}{C(r', q, d)}C(r, p, d)}\\
        &&=\frac{1-\alpha}{1+\alpha}\frac{1}{C(r, p, d)}\\
        &&=\left(1-\frac{2\alpha}{1+\alpha}\right)\frac{1}{C(r, p, d)}\\
        &&\geq \left(1-\frac{2\alpha}{1-\alpha}\right)\frac{1}{C(r, p, d)}.
    \end{eqnarray*}

    Therefore, the probability to generate a random lattice point in $B_d(r, p, d)$ is range of $$\left[(1-\alpha')\frac{1}{C(r, p, d)}, (1+\alpha')\frac{1}{C(r, p, d)}\right]$$ where $\alpha'=\frac{2\alpha}{1-\alpha}.$

    It takes $\bigO\left(\frac{d^3\log (r+\sqrt{d})}{\alpha}\right)$ running time to generate a random lattice point inside a $d-$dimensional ball $B_d(r+\sqrt{d}, p, d)$ with a lattice point center via Theorem~\ref{lattice-center-lattice-point}. Thus, the algorithm BigBallRandomLatticePoint(.) takes $\bigO\left(\frac{d^3\log (r+\sqrt{d})}{\alpha}\right)$ running time to generate a random lattice.
\end{proof}


\begin{theorem}\label{theorem43}
    For an arbitrary $\alpha\in (0, 1),$ there is an algorithm with runing time $\bigO\left(\frac{d^3\log (r+\sqrt{d})}{\alpha}\right)$ and $(1+\alpha)-$bias for a $d-$dimensional ball $B_d(r, q, d)$ to generate a random lattice pointo f radius $r>\frac{2d^{\frac{3}{2}}}{\alpha}$ with a arbitrary center; and there is a $\complexityLatticePoint$ time algorithm to generate a lattice point inside a $d-$dimensional ball $B_d(r, p, d)$ of radius $r\leq\frac{2d^{\frac{3}{2}}}{\alpha}$ with center $p\in D(\lambda,
    d, l)$.
\end{theorem}

\begin{proof}
    We discuss two cases based the radius of the $d$-dimensional ball.

    Case 1: When generate a random lattice point inside a $d$-dimensional ball of radius $r>\frac{2d^{\frac{3}{2}}}{\alpha}$ with center
    arbitrary center $p$, apply Theorem~ \ref{larger-radius-lattice-point}.

    Case 2: When generate a random lattice point inside a $d$-dimensional ball of radius $r\leq\frac{2d^{\frac{3}{2}}}{\alpha}$ with center $p\in D(\lambda,
    d, l)$, apply Theorem~ \ref{smaller-radius-lattice-point}.
\end{proof}

\subsection{Count Lattice Point in the Union of High Dimensional Balls}
In this section, we apply the algorithm developed in Section~\ref{alg} to count the total number of lattice point in the union of high dimensional balls.



\begin{theorem}
    There is a $\bigO\left(\poly\left({1\over \epsilon},\log{1\over \gamma}\right)\cdot
    {\setCount}\cdot (\log {\setCount})^{\bigO(1)}\right)$ time and
    $\bigO(\log \setCount)$ rounds algorithm for the number of lattice points in $B_1\cup B_2\cup\cdots\cup
    B_{\setCount}$  such that with probability at least $1-\gamma$, it
    gives a $sum\cdot M\in [{(1-\epsilon)(1-\alpha_L)(1-\beta_L)}\cdot
    |B_1\cup\cdots \cup B_{\setCount}|,
    {(1+\epsilon)(1+\alpha_R)(1+\beta_R)}\cdot |B_1\cup\cdots \cup
    B_{\setCount}|],$ where each ball $B_i$ satisfy that either its radius $r>\frac{2d^\frac{3}{2}}{\beta}$ or its center $p\in D(\lambda, d, l)$ and $|B_1\cup\cdots \cup
    B_{\setCount}|$ is the total number of lattice point of union of $\setCount$ high dimensional balls.
\end{theorem}
\begin{proof}
    Apply Theorem~\ref{theorem14} and Theorem~\ref{theorem43}, we have $m_i$ for each ball $B_i$ with $$m_i\in \left((1-\beta_L)C_i(r_i, p_i, t), (1+\beta_R)C_i(r_i, p_i, t)\right),$$ and biased random generators with $$\prob(x=\randomElm(B_i))\in
    \left[{1-\alpha_L\over C_i(r_i, p_i, t)},{1+\alpha_R\over C_i(r_i, p_i, t)}\right]$$ for each input
    ball $B_i$, where $C_i(r_i, p_i, t)$ is the number of lattice point of $t-$dimensional ball $B_i$ of radius $r_i$ for $i=1, 2, ..., \setCount$.
    Then apply Theorem~\ref{nonadaptive-thm}.
\end{proof}

\subsection{Hardness to Count Lattice Points in a Set of Balls}

In this section, we show that it is \#P-hard to count the number of
lattice points in a set of balls.

\begin{theorem}
    It is \#P-hard to count the number of lattice points in a set of
    $d$-dimensional balls even the centers are of the format
    $(x_1,\cdots, x_d)\in \mathbb{R}^d$ that has each $x_i$ to be either $1$ or
    ${\sqrt{h}\over 2}$ for some integer $h\le d$.
\end{theorem}

\begin{proof}
    We derive a polynomial time reduction from DNF problem to it. For
    each set of lattice points in a $h$-dimensional cube $\{0,1\}^h$, we
    design a ball with radius $r={\sqrt{h}\over 2}$ and center at
    $C=({\sqrt{h}\over 2},\cdots, {\sqrt{h}\over 2})$. It is easy to see
    that this ball only covers the lattice points in $\{0,1\}^h$. Every
    $0,1$-lattice point in ${0,1}$ has distance to the center $C$ equal
    to $r$. For every lattice point $P\in R^h$ that is not in
    $\{0,1\}^h$ has distance $d$ with $d^2\ge r^2+(1+{1\over
        2})^2-({1\over 2})^2= r^2+2$.
\end{proof}

\begin{definition}
    For a center $c=(c_1,\cdots, c_d)$ and an even number $k>0$ and a
    real $r>0$, a $d$-dimensional {\it $k$-degree} ball $B_k(c,r)$ is
    $\{(x_1,\cdots, x_d): (x_1,\cdots, x_d)\in \mathbb{R}^d$ and $\sum\limits_{i=1}^d
    (x_i-c_i)^k\le r\}$.
\end{definition}

\begin{theorem}\label{k-hard-thm}Let $k$ be an even number at least $2$. Then we have:
    \begin{enumerate}
        \item
        There is no polynomial time algorithm to approximate the number of
        lattice points in the intersection $n$-dimensional $k$-degree balls
        unless P=NP.
        \item
        It is \#P-hard to count the number of lattice points in the
        intersection $n$-dimensional $k$-degree  balls.
    \end{enumerate}
\end{theorem}

\begin{proof}
    We derive a polynomial time reduction from 3SAT problem to it. For
    each clause $C=(x_i^*\vee x_j^*\vee x_k^*)$, we can get a ball to
    contain all lattice points in the 0-1-cube to satisfy $C$, each
    $x_i^*$ is a literal to be either $x_i$ or its negation $\bar{x_i}$.

    Without loss of generality, let $C=(x_1\vee x_2\vee x_3)$. Let
    $\delta=0.30$. Let center $D_C=(d_1,d_2,d_3,d_4, d_5,\cdots,
    d_n)=(1-\delta,1-\delta,1-\delta,{1\over 2},{1\over 2},\cdots,
    {1\over 2})$, which has value $1-\delta$ in the first three
    positions, and ${1\over 2}$ in the rest. For $0,1$ assignment $(a_1,
    a_2, \cdots, a_n)$ of $n$ variables, if it satisfies $C$ if and only
    if $\sum\limits_{i=1}^n (a_i-d_i)^k\le 2(1-\delta)^k+\delta^k+(n-3)\cdot
    ({1\over 2})^k$. Therefore, we can select radius $r_C$ that
    satisfies $r_C^k=2(1-\delta)^k+\delta^k+(n-3)\cdot ({1\over 2})^k$.
    We have the following inequalities:
    \begin{equation} \label{first-out-ineqn}
    \left\{
    \begin{array}{lr}
    (2-\delta)^2>(1+\delta)^k> 2(1-\delta)^k+\delta^k\\
    (1+{1\over 2})^k>2(1-\delta)^k+\delta^k+({1\over
        2})^k.
    \end{array}
    \right.
    \end{equation}

    This is because we have the following equalities:
    \begin{equation} \label{first-out-ineqn}
        \left\{
        \begin{array}{lr}
            (1+\delta)^2=1.69,\\
            2(1-\delta)^2+\delta^2=2\times
            0.49+0.09=1.07,\\
            2(1-\delta)^2+\delta^2+({1\over
                2})^2=1.07+0.25=1.32,\\
            (1+{1\over 2})^2=2.25.
        \end{array}
        \right.
    \end{equation}


    If $Y=(y_1,y_2,\cdots, y_n)$ is not a $0,1$-lattice point, we
    discuss two cases:

    \begin{enumerate}
        \item
        Case 1. $y_i\not\in \{0,1\}$ for some $i$ with $1\le i\le 3$.

        In this case we know that dist$(Y, D_C)^2>r_C^2$ by
        inequality~(\ref{first-out-ineqn}).

        \item
        Case 2. $y_i\not\in \{0,1\}$ for some $i$ with $3< i\le n$.

        In this case we know that dist$(Y, D_C)^2>r_C^2$ by
        inequality~(\ref{first-out-ineqn}).
    \end{enumerate}

    If $Y=(y_1,y_2,\cdots, y_n)$ is a $0,1$-lattice point, we discuss
    two cases:

    \begin{enumerate}
        \item
        Case 1. $Y$ satisfies $C$.

        In this case we know that dist$(Y, D_C)^2\le r_C^2$.

        \item
        Case 2. $Y$ does not satisfy $C$.

        In this case we know that dist$(Y, D_C)^2>r_C^2$ by inequality
        $(1-\delta)^2>\delta^2$.
    \end{enumerate}

    The ball $B_C$ with center at $D_C$ and radius $r_C$ contains
    exactly those 0,1-lattice points that satisfy clause $C$. This
    proves the first part of the theorem.

    If there
    were any factor $c$-approximation to the intersection of balls, it
    would be able to test if the intersection is empty. This would bring
    a polynomial time solution to 3SAT.

    It is well known that \#3SAT is \#P-hard. Therefore, It is \#P-hard
    to count the number of lattice points in the intersection
    $n$-dimensional balls. This proves the second part of the theorem.

\end{proof}




\section{Approximation for the Maximal Coverage
    with Balls}\label{ser}
We apply the technology developed in this paper to the maximal
coverage problem when each set is a set of lattice points in a ball
with center in $D(\lambda,d, l)$.

The classical maximum coverage is that given a list of sets
$A_1,\cdots, A_{\setCount}$ and an integer $k$, find $k$ sets from
$A_1, A_2, \cdots, A_{\setCount}$ to maximize the size of the union
of the selected sets in the computational model defined in
Definition~\ref{input-list-def}. For real number $a\in [0,1]$, an
approximation algorithm is a $(1-a)$-approximation for the maximum
coverage problem that has input of integer parameter $k$ and a list
of sets $A_1,\cdots, A_{\setCount}$ if it outputs a sublist of sets
$A_{i_1}, A_{i_2},\cdots, A_{i_k}$ such that $|A_{i_1}\cup
A_{i_2}\cup \cdots\cup A_{i_k}|\ge (1-a) |A_{j_1}\cup A_{j_2}\cup
\cdots\cup A_{j_k}|$, where $A_{j_1}, A_{j_2}, \cdots, A_{j_k}$ is a
solution with maximum size of union.

\begin{theorem}\label{appr3-cover-theorem}\cite{fu2016} Let $\rho$ be a constant in
    $(0,1)$. For  parameters $\xi,\gamma\in (0,1)$ and
    $\alpha_L,\alpha_R,\delta_L,\delta_R\in [0,1-\rho]$, there is an
    algorithm to give a $\left(1-(1-{\beta\over k})^k-\xi\right)$-approximation for
    the maximum cover problem, such that given a
    $((\alpha_l,\alpha_r),(\delta_L,\delta_R))$-list $L$ of finite sets
    $A_1,\cdots, A_{m}$ and an integer $k$,  with probability at
    least $1-\gamma$, it returns an integer $z$ and a subset $H\subseteq
    \{1,2,\cdots, m\}$ that satisfy
    \begin{enumerate}
        \item
        $|\cup_{j\in H}A_j|\ge \left(1-(1-{\beta\over k})^k-\xi\right)C^*(L,k)$ and
        $|H|=k$,
        \item
        $((1-\alpha_L)(1-\delta_L)-\xi)|\cup_{j\in H}A_j|\le z\le
        ((1+\alpha_R)(1+\delta_R)+\xi)|\cup_{j\in H}A_j|$, and
        \item
        Its complexity is  $(T(\xi,\gamma,k, \setCount), R(\xi,\gamma
        ,k,\setCount), Q(\xi,\gamma ,k,\setCount))$ with
        \begin{eqnarray*}
            T(\xi,\gamma,k,\setCount)&=&\bigO\left({k^3\over \xi^2}\left(k\log\left({3
                \setCount\over k}\right)+\log
            {1\over \gamma}\right)\setCount\right),
        \end{eqnarray*}
        where $\beta={(1-\alpha_L)(1-\delta_L)\over
            (1+\alpha_R)(1+\delta_R)}$ and $C^*(L,k)$ is the number of elements to be covered in an optimal solution.
    \end{enumerate}
\end{theorem}

We need Lemma~\ref{help-app-lemma} to transform the approximation
ratio given by Theorem~\ref{appr3-cover-theorem} to constant
$(1-{1\over e})$ to match the classical ratio for the maximum
coverage problem.

\begin{lemma}\label{help-app-lemma}
    For each integer $k\ge 2$, and real $b\in [0,1]$, we have:
    \begin{enumerate}
        \item\label{first-help-app-lemma}
        $(1-{b\over k})^k\le {1\over e}-{\eta\over e} (b+{b \over 2k}-1)$.
        \item\label{secon-help-app-lemma}
        If $\xi\le {\eta\over e} (b+{b \over 2k}-1)$, then
        $1-(1-{b\over k})^k-\xi> 1-{1\over
            e}$,  where $\eta=e^{-{1\over 4}}$.
    \end{enumerate}
\end{lemma}

\begin{proof}
    Let function $f(x)=1-\eta x-e^{-x}$. We have $f(0)=0$. Taking
    differentiation, we get ${d f(x)\over dx}=-\eta +e^{-x}>0$ for all
    $x\in (0,{1\over 4})$.

    Therefore, for all $x\in (0,{1\over 4})$,
    \begin{eqnarray}
    e^{-x}\le 1-\eta x.\label{exp-expansion-ineqn}
    \end{eqnarray}
    The following Taylor expansion can be found in standard calculus
    textbooks. For all $x\in (0,1)$,
    \begin{eqnarray*}
    \ln (1-x)=-x-{x^2\over 2}-{x^3\over 3}-\cdots .
    \end{eqnarray*}

    Therefore, we have
    \begin{eqnarray}
    (1-{b\over k})^k&=&e^{k\ln (1-{b\over k})}=e^{k(-{b\over
            k}-{b^2\over 2k^2}-{b^3\over 3k^3}-\cdots)}
    =e^{-b-{b^2\over 2k}-{b^3\over 3k^2}-\cdots}\nonumber\\
    &\le&e^{-b-{b\over 2k}}=e^{-1}\cdot e^{1-b-{b\over 2k}}\label{maximal-cov-ineqn1}\\
    &\le& e^{-1}\cdot (1-\eta \cdot (b+{b\over 2k}-1))\le{1\over
        e}-{\eta\over e} (b+{b \over 2k}-1)\label{maximal-cov-ineqn2}.
    \end{eqnarray}
    Note that the transition from~(\ref{maximal-cov-ineqn1})
    to~(\ref{maximal-cov-ineqn2}) is based on
    inequality~(\ref{exp-expansion-ineqn}).

    The part~\ref{secon-help-app-lemma} follows from
    part~\ref{first-help-app-lemma}. This is because $1-(1-{b\over
        k})^k-\xi\ge 1-{1\over e}+{\eta\over e} (b+{b \over 2k}-1)-\xi\ge
    1-{1\over e}$.
\end{proof}

\begin{theorem}
    There is a poly$(\lambda,d, l, k, \setCount)$
    time  $(1-{1\over e})$-approximation algorithm for maximal coverage
    problem when each set is the set of lattice points in a ball with
    center in $D(\lambda,d, l)$.
\end{theorem}
\begin{proof}[Sketch]
    Let $\alpha=\alpha_L=\alpha_R=\delta_L=\delta_R={1\over ck}$ with
    $c=100$, and $b=\beta={1-\alpha_L)(1-\delta_L)\over
        (1+\alpha_R)(1+\delta_R)}$. It is easy to see $(b+{b \over 2k}-1)\ge
    {1\over 4k}$. Let $\xi= {\eta\over e} (b+{b \over
        2k}-1)=\Theta({1\over k})$.
    It follows from Theorem~\ref{appr3-cover-theorem}, Lemma~\ref{help-app-lemma}, Theorem~\ref{theorem14} and Theorem~\ref{theorem43}.
\end{proof}

\section{Conclusions}\label{concl}
We introduce an almost linear bounded rounds randomized approximation algorithm for the size of set union problem $\arrowvert A_1\cup A_2\cup...\cup A_{\setCount}\arrowvert$, which given a list of sets $A_1,...,A_\setCount$ with approximate set size and biased random generators. The definition of round is introduced.
We prove that our algorithm runs sublinear in time under certain condition. A polynomial time approximation scheme is proposed to approximae the number of lattice points in the union of d-dimensional ball if each ball center satisfy $D(\lambda, d, l)$. We prove that it is $\#$P-hard to count the number of lattice points in a set of balls, and we also show that there is no polynomial time algorithm to approximate the number of lattice points in the intersection of $n$-dimenisonal $k$-degree balls unless P=NP.

\section{Acknowledgements}

We want to thank Peter Shor, Emil Je$\check{r}\acute{a}$bek, Rahul Savani et
al. for their comments about algorithm to geneate a random grid
point inside a $d-$dimensional ball on Theoretical Computer Science
Stack Exchange.






\end{document}